\documentclass[dvips,a4paper]{article}

\usepackage{cite}

   \usepackage[dvips]{graphicx,color} 
\graphicspath{%
    {Figures/pdf/}
    {converted_graphics/}
}

\DeclareGraphicsExtensions{.pdf}
\usepackage{graphicx}
\usepackage[cmex10]{amsmath}

\usepackage[caption=false,font=footnotesize]{subfig}

\usepackage{amsfonts}
\usepackage[caption=false,font=footnotesize]{subfig}
\newtheorem{proposition}{Proposition}
\newtheorem{proof}{Proof}

\numberwithin{equation}{section} 
\def\x{{\mathbf x}}
\def\y{{\mathbf y}}
\def\n{{\mathbf n}}
\def\w{{\mathbf w}}
\def\A{{\mathbf A}}
\def\Q{{\mathbf Q}}
\def\I{{\mathbf I}}
\def\J{{\mathbf J}}
\def\L{{\mathbf L}}
\def\D{{\mathbf D}}
\def\O{{\mathbf O}}
\def\1{{\mathbf1}}
\def\0{{\mathbf0}}
\def\X{{\mathbf X}}

\begin{document}

\title{ Compressed Sensing Performance Analysis via Replica Method using Bayesian framework}

\author{Solomon A. Tesfamicael   \and  Bruhtesfa E. Godana  \and  Faraz Barzideh }

%

%

\markboth{DRAFT- for internal use only}%
{Solomon A.Tesfamicael (tesfamic@iet.ntnu.no -\hskip 1mm Spring 2010}
\maketitle
\begin{abstract}
 
Compressive sensing (CS) is a new methodology to capture signals at lower rate than the Nyquist sampling rate when the signals are sparse or sparse in some domain. The performance of CS estimators is analyzed in this paper using tools from statistical mechanics, especially called replica method. This method has been used to analyze communication systems like Code Division Multiple Access (CDMA) and multiple input multiple output (MIMO) systems with large size. Replica analysis, now days rigorously proved, is an efficient tool to analyze large systems in general. Specifically, we analyze the performance of some of the estimators used in CS like LASSO (the Least Absolute Shrinkage and Selection Operator) estimator and Zero-Norm regularizing estimator as a special case of maximum a posteriori (MAP) estimator by using Bayesian framework to connect the CS estimators and replica method. We use both replica symmetric (RS) ansatz and one-step replica symmetry breaking (1RSB) ansatz, clamming the latter is efficient when the problem is not convex. This work is more analytical in its form. It is deferred for next step to focus on the numerical results.

\end{abstract}

\section{Introduction}\label{sec: Introduction}

Recently questions like, \textit{why go to so much effort to acquire all the data when most of what we get will be thrown away?};\textit{Can we not just directly measure the part that will not end up being thrown away?}, that were paused by Donoho\cite{CS} and others, triggered  a new way of sampling or sensing  called compact ("compressed") sensing (CS).

 In CS the task is to estimate or recover a sparse or compressible vector $\x^{0} \in \mathbb{R}^{N}$ from  a measurement vector $\y \in \mathbb{R}^{M}$. These are related through the linear transform  $ \y= \A\x^{0}$.
Here, $\x^{0}$ is a sparse vector and $M\ll N$. In the seminal papers\cite{CS} -\cite{CT}, $\x^{0}$ is estimated from $\y$, by solving a convex optimization problem \cite{SBB},\cite{FNW}. Others have used greedy algorithms, like subspace pursuit (SP)\cite{Dai}, orthogonal matching pursuit (OMP) \cite{Tropp}  to solve the problem. In this paper the focus is rather on the convex optimization methods. And we consider the noisy measurement system and the linear relation becomes

 \begin{equation}\label{noisy CS problem}
        \y= \A\x^{0} + \sigma_{0}\w.
   \end{equation}
Here, $\y$ and $\x^{0}$ are as in above where as  the noise term, $\w \sim \mathcal{N}(0,\I)$. There exists a large body of work on how to efficiently obtain an estimate for $\x^{0}$. And the performances of such estimators are measured using metrics like Restricted Isometric Property (RIP) \cite{CT2}, Mutual Coherence (MC) \cite{DCS}, yet there is apparently no consensus on the bounds in using such metrics. The tool used in this paper gives performance bounds of large size CS systems \cite{GBS}.  \\

Generally the linear model $\eqref{noisy CS problem}$ is used to describe a multitude of linear systems like code division multiple access (CDMA) and multiple antenna systems like MIMO, to mention just a few. Tools from statistical mechanics have been employed to analyze large CDMA \cite{Tanaka}  and MIMO systems \cite{Ralf} \cite{RALF2}, and on in this paper  the same wisdom is applied to analyze the performance of estimators used in CS. Guo and et al in \cite{GBS} used a Bayesian framework for statistical inference with noisy measurements and characterize the posterior distribution of individual elements of the sparse signal by describing the mean mean square error(MSE) exactly. To do so, they consider $\eqref{noisy CS problem}$ in a large system  and applied the decoupling principle using tools from statistical mechanics.\\

One can find also works that have used the tools from statistical mechanics to analyse CS system performances. To mention some,  in \cite{GBS} as stated above, Guo and et al used the tools to describe the minimum mean square error  (MMSE) estimator, in  \cite{MAP} Rangan and others used the maximum a posterior(MAP) estimator of CS systems. These are referred as Replica MMSE claim and Replica MAP claim in \cite{MAP}.

In \cite{SBG} -\cite{Ari} authors have used Belief propagation and message passing algorithms for probabilistic reconstruction in CS using replica methods including RS. Especially, in \cite{Krz} one finds excellent work about phase diagrams in CS systems while \cite{Tul} generalizes replica analysis using free random matrices. Kabashima and et. al in \cite{KWT}, Ganguli and Sompolinsky in \cite{SS}  and Takeda and Kabashima \cite{TK} -\cite{THK} have shown statistical mechanical analysis of the CS by considering the noiseless recovery problem and they indicated that RSB analysis is needed in the phase regimes where the RS solution is not stable. In this paper  the performance of those CS estimators, considered as MAP estimator, is shown for the noisy problem by using the replica method including RS and RSB as in \cite{RGM} -\cite{BRAR}, where the RSB ansatz gives better solution when the replica symmetry (RS) solution is unstable. This work is kind of an extension of  \cite{BRAR} from MIMO systems to the CS systems.\\



The paper is organized as follows.  In section \ref{sec:Problem setting } the estimator in CS system are presented and redefined using  the Bayesian framework, and based on that we present our basis of analysis in section \ref{sec: Analysis} which is the replica method from the statistical physics  and apply it on the different CS estimators which are presented generally as a MAP estimator.  In section \ref{sec: Particular Example} we showed our analysis using a paricular example, and section \ref{sec: Conclusion} presents conclusion and of future work.


\section{Bayesian framework for Sparse Estimation}\label{sec:Problem setting }

Beginning with a given vector of  measurements $\y \in \mathbb{R}^{M}$  and measurement matrix $\A \in \mathbb{R}^{M \times N}$,  assuming noisy  measurement with $\w  \in \mathbb{R}^{M} $  being i.i.d. Gaussian random variables with zero mean and covariance matrix  $\I $, estimating the sparse vector $\x^{0}  \in \mathbb{R}^{N}$ is the problem that we are considering  where these variables are related by the linear model \eqref{noisy CS problem}.


\subsection{Sparse Signal Estimation}\label{sec:CS}

Various methods for estimating $\x^{0}$ may be used. The classical approach to solving inverse problems of such type is by least squares (LS) estimator in which no prior information is used and its closed form is
\begin{equation}\label{LS}
         \hat{\x^{0}} = ( \A^{T} \A)^{-1}\A^{T}\y,
   \end{equation}
which performs very badly for the CS estimation problem we are considering since it does not find the sparse solution. Another approach to estimate  $\x^{0}$  is via the solution of the unconstrained optimization problem 

\begin{equation}\label{unconstrained optimization problem}
 \hat{\x^{0}} =\underset{ \x \in \mathbb{R}^{N}}{\operatorname{\mbox{min}}} \frac {1}{2} \parallel \y - \A \x^{0} \parallel_{2}^{2}   + u  f( \x^{0}), 
 \end{equation}
where $u  f( \x^{0} )$ is a regularizing term, for some non-negative $u$. By taking $f(\x^{0})=\parallel \x^{0} \parallel_{p}$, emphasis is made on a solution with LP norm, and  $\parallel \x^{0}\parallel_{p}$  is defined as a penalizing norm.  When $p=2$, we get

\begin{equation}\label{L2}
 \hat{\x^{0}} =\underset{ \x^{0} \in \mathbb{R}^{N}}{\operatorname{\mbox{min}}} \frac {1}{2} \parallel \y - \A \x^{0} \parallel_{2}^{2}   + u  \parallel \x^{0} \parallel_{2}.
 \end{equation}
This is penalizing the least square error by the  L2 norm  and this performs badly as well, since it does not introduce sparsity into the problem. When $p=0$, we get the L0 norm, which is defined as 

\begin{equation}\label{L0def}
  \nonumber  \| {\x^{0}} \|_{0} =k\equiv  \# \bigl \{ i \in \{1,2, \cdots ,N \} | x^{0}_{i}  \neq 0\bigr \} ,
 \end{equation}
the number of the non zero intries of  $\x^{0}$, which actually is a partial norm since it does not satisfy the triangle inequality property, but can be treated as norm by defining it as in \cite{MAP}, and get the L0 norm regularizing estimator 
\begin{equation}\label{L0}
 \hat{\x^{0}} =\underset{ \x^{0} \in \mathbb{R}^{N}}{\operatorname{\mbox{min}}} \frac {1}{2} \parallel \y - \A \x^{0} \parallel_{2}^{2}   + u  \parallel \x^{0} \parallel_{0},
 \end{equation}
which gives the best solution for the problem at hand since it favors sparsity in $\x^{0}$. Nonetheless, it is an NP- hard combinatorial problem. Instead, it has been a practice to approximate it using L1  penalizing norm to get the estimator 
\begin{equation}\label{L1}
 \hat{\x^{0}} =\underset{ \x^{0} \in \mathbb{R}^{N}}{\operatorname{\mbox{min}}} \frac {1}{2} \parallel \y - \A \x^{0} \parallel_{2}^{2}   + u  \parallel \x^{0} \parallel_{1},
 \end{equation}
which is a convex approximation to the L0 penalizing solution \ref{L0}.  The best solution for estimate of the sparse vector $\x $  is given by the zero-norm regularized estimator which  is a hard combinatorial problem. These estimators, \eqref{L2} - \eqref{L1}, can equivalently be presented as solutions to constrained optimization problem \cite{CS}-\cite{CT}. This constrained optimization version of \eqref{L1} is known as the L1 penalized L2 minimization called LASSO (Least Absolute Shrinkage and Selection Operator) or BPDN(Basis Persuit Denoising), which can be set as Quadratic Programing (QP) and Quadratic Constrained Linear Programing (QCPL) optimization problems. \footnote[1]{In this paper we consider the former and leave the later as they are equivalent algorithms.} In the following subsection the above estimators are presented as a MAP estimator in Bayesian framework.

\subsection{Bayesian framework for Sparse signal}\label{sec:BCS}


Equivalently, the estimator of $\x^{0} $  in \eqref{unconstrained optimization problem} can generally be presented as MAP estimator under the Bayesian framework. Assume a prior probability distribution for $\x$ to be 
\begin{equation}\label{prior}
p_{u}(\x)= \frac {e^{-u f(\x)}}{\int_{\x \in \chi^{N}}e^{-u f(\x)}d\x}, 
\end{equation}
 where the cost function $ f : \chi \rightarrow \mathbb{R} $ is some scalar-valued, non negative function with $\chi \subseteq \mathbb{R} $ and
  \begin{equation}
f(\x)= \sum_{i=1}^{N}  f(x_{i}).
\end{equation}
such that for sufficiently large $u$,  $\int_{\bf{x} \in \chi^{n}} \exp (-u f(\x)) d\x$ is finite as in \cite{MAP}. And let the assumed variance of the noise  be given by
$$\sigma_{u}^2=\frac{\gamma}{u} $$ where $ \gamma $ is system parameter which can be taken as  $ \gamma=\sigma_{u}^2u $ where $\sigma_{u}^2$ is the assumed variance for each component of $\bf{n}$. Note that we incorporate  the sparsity in the prior pdf via $f(\bf{x})$. By \eqref{noisy CS problem} the probability density function of $\bf{y}$ given  $\bf{x}$  is given by 
 \begin{equation}
p_{\y\mid\x}(\y\mid\x;\A)= \frac{1}{(2\pi \sigma_{u}^2)^{N/2}}
e^{-\frac {1}{2  \sigma_{u}^2} \parallel \y - \bf{A} \x \parallel_{2}^{2}   } , 
\end{equation}
and prior distribution of $\bf{x}$  by \eqref{prior},  the posterior distribution  for the measurement channel \eqref{noisy CS problem} according to Bayes law is  
 \begin{equation}\label{posterior}
p_{\x\mid\y}(\x\mid\y;\A)=\frac {e^{-u(\frac {1}{2 \gamma} \parallel \y - \A \x \parallel_{2}^{2}   +   f(\x))}}{\int_{\x \in \chi^{n}}e^{-u(\frac {1}{2 \gamma} \parallel \y - \A \x \parallel_{2}^{2}   +   f(\x))}d\x}.
\end{equation}

Then the MAP estimator can be shown to be
\begin{equation}\label{MAP}
\hat{\x}^{MAP}= \underset{\x \in \chi^{n}}{\operatorname{arg min }}   \frac {1}{2\gamma} \parallel \y - \A \x \parallel_{2}^{2}   +  f(\x).
\end{equation} 
Now, as we choose different penalizing function in (\ref{MAP}) we get the different estimators defined above in equations \eqref{L2}, \eqref{L0}, and \eqref{L1} but this time under the Bayesian framework as a MAP estimator \cite{MAP}.
\begin{enumerate}
\item
 Linear Estimators: when  $f(\x)=\parallel  \x\parallel_{2}^{2}$  \eqref{MAP} reduces to 
\begin{equation}\label{Linear_MAP}
\hat{\x}_{Linear}^{MAP}=\A^{T}(\A \A^{T} +\gamma\I)^{-1}\y, 
\end{equation}
which is the LMMSE estimator. 
\item 
LASSO Estimator: when $f(\x)=\parallel  \x \parallel_{1}$ we get the LASSO estimator and \eqref{MAP} becomes
\begin{equation}\label{LASSO_MAP}
\hat{\x}_{Lasso}^{MAP}= \underset{\x \in \chi^{n}}{\operatorname{arg min }} \frac {1}{2\gamma} \parallel \y - \A \x \parallel_{2}^{2}   +   \parallel  \x \parallel_{1}.
\end{equation} 
\item
 Zero-Norm regularization estimator:  when $ f({\x})=\| {\x} \|_{0} $ ,  we get  the  Zero-Norm regularization estimator and \eqref{MAP} becomes
\begin{equation}\label{Zero Norm_MAP}
\hat{\x}_{Zero}^{MAP}= \underset{\x \in \chi^{n}}{\operatorname{arg min }} \frac {1}{2\gamma} \parallel \y- \A \x \parallel_{2}^{2}   +   \parallel  \x \parallel_{0}.
\end{equation} 
\end{enumerate}

Whether these minimization problems are solvable or not the replica analysis results can provide the asymptotic performances of all the above estimators via replica method as showed  in  \cite{GBS},  \cite{MAP}, \cite{KWT}, \cite{SS} and \cite{TK}. We apply  RS ansatz as used by M\"{u}ller and et al in \cite{RGM}  and RSB ansatz as used by Zaidel and et al  \cite{BRAR}  on vector precoding for MIMO. Actually, this work is an extension of the RSB analysis to MIMO systems done in \cite{BRAR} to the CS system. 



\section{A Statistical Physics Analysis  }\label{sec: Analysis}

The performance of the Bayesian estimators like MMSE and MAP can be done using the pdf of the error vector. The error is random and it should be centered about zero for the estimator to perform well.  Kay showed  in that way (see section 11.6 in \cite{KAY})  the performance analysis of MMSE estimator.  We believe in general that inference for the asymptotic performance of MAP estimators is best done with statistical mechanical tools including RSB assumption and this is done in the sense of the mean square error (MSE). 

The posterior distribution \eqref{posterior} is  a sufficient statistics to estimate $\x^{0}$  \cite{GBS} and the denominator is called the normalizing factor or evidence in Bayesian inference according to \cite{macKAY} and Partition function in statistical mechanics. Actually, it is this connection, which gives the ground to apply the tools, which are used in statistical mechanics.  So the task of evaluating the above estimators for the sparse vector $\x^{0} $ can be translated to the statistical physics framework. And let us justify first how the analysis using statistical mechanical tool is  able to do it. \\   \\Define the Gibbs-Boltzmann distribution as 

\begin{equation} \label{definition of Gibbs-Boltzman dist.}
p_{\x}(\x)= \frac {1}{ \mathcal{Z}}e^{-\beta \cal {H}(\x)}
\end{equation}
where $ \beta $  is a constant known as the  inverse temperature in
the terminology of physical systems. For small $ \beta $, the prior probability becomes
flat, and for large $ \beta $, the prior probability has sharp modes. $\cal H$, which is an expression of the total energy of the system, is called the Hamiltonian in physics literature and   $\mathcal{Z} $ is the partition function given by

 \begin{equation}\label{definition of partition function}
\mathcal{Z}= \sum_{\chi^{N}}e^{-\beta \cal{ H}(\x)}d\x. 
\end{equation}
 Often the Hamiltonian can be given by a quadratic form like 

\begin{equation}\label{definition of Hamiltonian}
\mbox{$\cal {H}$} (\x)=\x^T \J \x,
\end{equation} 
with $ \J$ being a Random matrix of dimension $N \times N$.  Then the minimum average energy per component of $ \x$  can be given by 
\begin{equation}\label{eq:definition of energy}
\mathcal{ E }= \frac {1}{N}  \hspace{2mm}   \underset{\bf{x} \in  \chi^N }{\operatorname{min }} \hspace{2mm}   \cal H (\bf{x})
 \end{equation}
For our system that we considered to address, which is given by \eqref{MAP}  or equivalently by  \eqref{unconstrained optimization problem}, the Hamiltonian becomes
\begin{equation}\label{Hamiltonian  the present system}
\mbox{$\cal {H}$} (\x)= \frac{1}{2 \sigma_{u}^2 } (\y - \A \x )^{T}  (\y- \A \x )  + u f(\x) .
\end{equation}
 Compared to \eqref{definition of Hamiltonian},  the Hamiltonian in  \eqref{Hamiltonian  the present system} has regularizing term in addition to the quadratic form, which is the energy of the error,  in which the regularizing term $ f(\x)$ is accountable for addressing the problem in the CS. The Gibbs-Boltzman distribution is a solution to \eqref{MAP}  or to \eqref{unconstrained optimization problem} in general, after plugging \eqref{definition of partition function} and \eqref{Hamiltonian  the present system} since they are equivalent problems. The normalizing factor ( aslo called the partition function) of this distribution is central for calculating many important  variables and we shall begin from this term to analyse the CS estimators performance. \\

Assuming that $\x^{0}$ and $\x$ being drawn from the same discrete set (we shall later provide an example from such a set).  The partition function of the posterior distribution given in \eqref{definition of Gibbs-Boltzman dist.}  
%
%
becomes
 \begin{equation}\label{the partition function of  the present system}
\mathcal{Z}=\sum\limits_{\x \in \chi^{N}}e^{-\beta \big[\frac {1}{2\sigma_{u}^2} \parallel \y - \A \x \parallel_{2}^{2}   + u f(\x)\big]}, 
\end{equation}
by using  \eqref{definition of partition function} and \eqref{Hamiltonian  the present system}.  The posterior distribution \eqref{MAP} depends on the predetermined random variables $\y$ and $\A$ called quenched states in physics literature \cite{TUK}, \cite{THK}. That is, we use fixed states $ \y = \A \x^{0} +\w $  instead of $\y $ for the large system limit, as $ N, M \rightarrow \infty$, while maintaining $N/M$ fixed. We then calculate the nth moment of the partition function $Z$ with respect to the predetermined variables, $n$ replicas, hence the name replica method came from. The replicated partition function is then given by 
  \begin{equation}
\mathcal{Z}^n=\sum\limits_{\{\x^{a}\}}e^{-\beta \Bigl[ \frac {1}{2\sigma_{u}^2} \sum\limits_{a=1}^{n} \bigl (\parallel \bf{y} - \A \x^{a} \parallel_{2}^{2} \bigr)  +   \frac{\gamma} {\sigma_{u}^2}  \sum\limits_{a=1}^{n}  f(\x^{a}) \Bigr]},
\end{equation}
where $\sum\limits_{\{\x^{a}\}}=\sum\limits_{\x_{1} \in \chi^{N}}...\sum\limits_{\x_{n} \in \chi^{N}}$. And after substituting $\y$, it becomes 
   \begin{equation}\label{ the replicated  partition function }
\mathcal{Z}^n=\sum\limits_{\{\x^{a}\}}e^{-\beta \Bigl[ \frac {1}{2\sigma_{u}^2} \sum\limits_{a=1}^{n} \bigl (\parallel \A(  \x^{0}  - \x^{a})+\w \parallel_{2}^{2} \bigr)  +   \frac{\gamma} {\sigma_{u}^2}  \sum\limits_{a=1}^{n}  f(\x^{a}) \Bigr]} . 
\end{equation}

Averaging over the noise $\n$ first, we get     
\begin{align}\label{ averaged, over noise, replicated  partition function }
\int_{\mathbb{R}^M}\frac{d\n}{\pi^M}e^{-\frac {1}{2\sigma_{0}}( \w^T \w ) }\mathcal{Z}^n =
\sum\limits_{\{\bf{x}^{a}\}}e^{-\beta \Bigl[ \frac {1}{2} Tr \J \L(n)  +  \frac {\gamma} {\sigma_{u}^2}  \sum\limits_{a=1}^{n}  f(\bf{x}^{a}) \Bigr]}, 
 \end{align}
where $\J=\A^T\A $ and it is assumed to decompose into 
 \begin{equation} \label{eq:decomposiblity property}
\J=\O\D\O^{-1},
\end{equation}
and $\D$ is a diagonal matrix while $\O$ is $N \times N$ orthogonal matrix assumed to be drawn randomly from the uniform distribution defined by the Haar measure on the orthogonal group. For more clarity on this one can see --- in  \cite{TUK}.
 And $L(n)$ is given by 
\begin{equation}
   \scriptstyle{\L(n)=-\frac{1}{\sigma_{u}^2}\sum\limits_{a=1}^{n}(  \x^{0}  - \x^{a})(  \x^{0}  - \x^{a})^T+ \frac{\sigma_{0}^2}{\sigma_{u}^2(\sigma_{u}^2+n\sigma_{0}^2)}\Biggl(\sum\limits_{a=1}^{n}(  \x^{0}  - \x^{a})\Biggr) \Biggl(\sum\limits_{b=1}^{n}(  \x^{0}  - \x^{b})\Biggr)^T}.
\end{equation}
Further averaging what we get on the right side of \eqref{ averaged, over noise, replicated  partition function } over the cross correlation matrix $\J$, by assuming the eigenvalue spectrum of $\J$ to be self-averaging, we get 
 \begin{align}\label{the Harish -Chandra -Itzykoson-Zuber integeral}
\nonumber \underset{\w,\J}   {\operatorname{E }} \bigl \{ \mathcal{Z}^n \bigr \}& = \underset{\bf{J}}   {\operatorname{E }} \Biggl( \sum\limits_{\{\x^{a}\}}e^{-\beta \Bigl[ \frac {1}{2} Tr \J\L(n)  +  \frac {\gamma} {\sigma_{u}^2}  \sum\limits_{a=1}^{n}  f(\x^{a}) \Bigr]} \Biggr) \\&=  \sum\limits_{\{\x^{a}\}}e^{  \frac {-\beta\gamma} {\sigma_{u}^2}  \sum\limits_{a=1}^{n}  f(\x^{a})} \underset{\J}   {\operatorname{E }} \Biggl(e^{-\beta \Bigl[ \frac {1}{2} Tr \J \L(n)   \Bigr]} \Biggr)
 \end{align}

The inner expectation in \eqref{the Harish -Chandra -Itzykoson-Zuber integeral} is the Harish -Chandra -Itzykoson-Zuber integral (again see in  \cite{RGM} and \cite{BRAR} and the references therein).  The plan here is to evaluate the fixed-rank matrices $\L(n)$ as $N\rightarrow \infty$.  Further following the explanation in  \cite{BRAR} \eqref{the Harish -Chandra -Itzykoson-Zuber integeral}   becomes 
\begin{equation}\label{R-transform plus zero order}
\underset{\w,\J }   {\operatorname{E }} \bigl \{ \mathcal{Z}^n \bigr \} = \sum\limits_{\{\x^{a}\}}e^{  \frac {-\beta\gamma} {\sigma_{u}^2}  \sum\limits_{a=1}^{n}  f(\x^{a})} e^{-N \sum\limits_{a=1}^{n}\int_{0}^{\lambda_{a}}R(-w)dw + o(N)} 
\end{equation}
 where  $R(w)$ is the R-transform of the limiting eigenvalue distribution  of the matrix J( see, definition 1 in  \cite{RGM}  of R-transform  or in \cite{Ralf} and   \cite{RALF2} for better understanding of R-transform)  and $\{\lambda_a \}$ denote the Eigenvalues of the  $ n \times n $  matrix $    \beta\Q$ , with $\Q$ defined through
  \begin{equation}\label{Correlation matrix}
 \scriptstyle{Q_{ab}\equiv  \frac{1}{N} \Biggl[-\frac{1}{\sigma_{u}^2}\sum\limits_{i=1}^{N}   (x_{i}^{0}  -  x_{i}^{a})^T(  x_{i}^{0}  -  x_{i}^{b}) + \frac{\sigma_{0}^2}{\sigma_{u}^2(\sigma_{u}^2+n\sigma_{0}^2)}\Biggl(\sum\limits_{i=1}^{N} (  x_{i}^{0}  -  x_{i}^{a})\Biggr)^T \Biggl(\sum\limits_{i=1}^{N} (  x_{i}^{0}  -  x_{i}^{b})\Biggr)\Biggr]},
\end{equation}
 for   $a,b=1, \cdots , n$.

After applying replica trick, the average free energy can be given by 
 \begin{align}  \label{eq:average free energy}
\nonumber\beta\bar{\mathcal{F}}&=-  \underset{N \rightarrow  \infty}{\operatorname{lim }} \frac{1}{ N}   \underset{\bf{n},\bf{R}}   {\operatorname{E }} \{\log\hspace{1mm} \mathcal{Z} \}\\
& = -  \underset{N \rightarrow  \infty}{\operatorname{lim }} \frac{1}{ N}  \underset{n \rightarrow 0}{\operatorname{lim }}     \frac{\partial}{ \partial n}  \log \underset{\bf{n}, \bf{R} }{\operatorname{E }}\{ ( \mathcal{Z} )^{n}\}  \\
 \nonumber
 \end{align}
and the energy of the error can be calculated from the average free energy as 
\begin{align}\label{avarage energy penality}
\bar{\mathcal{E}}&= \underset{\beta \rightarrow  \infty}{\operatorname{lim }} \frac{1}{ \beta}\bar{\mathcal{F}}\\
&=- \underset{\beta \rightarrow  \infty}{\operatorname{lim }} \frac{1}{ \beta} \underset{N \rightarrow  \infty}{\operatorname{lim }} \frac{1}{ N}   \underset{\bf{n},\bf{R}}   {\operatorname{E }} \{\log\hspace{1mm} \mathcal{Z} \}  \nonumber \\&=-\underset{\beta \rightarrow  \infty}{\operatorname{lim }} \frac{1}{ \beta}\underset{n \rightarrow 0}{\operatorname{lim }}    \frac{\partial}{ \partial n} \underbrace{  {\rm   \underset{N \rightarrow  \infty}{\operatorname{lim }} \frac{1}{ N}  \log  \underset{\bf{n}, \bf{J} }{\operatorname{E }}\{ ( \mathcal{Z} )^{n}\} } }_{\rm  \Xi_{n}} .
\end{align}
where we get \eqref{avarage energy penality} by using one of the assumptions used in replica calculations, after interchanging the order of the limits  we assumed we get the same result. Further, for $\Xi_{n}$ we have 
\begin{equation}\label{eq:R-transform without zero order} 
\Xi_{n}=-  \underset{N \rightarrow  \infty}{\operatorname{lim }} \frac{1}{N} \log \Biggl(\sum\limits_{\{\x^{a}\}}e^{  \frac {-\beta\gamma} {\sigma_{u}^2}  \sum\limits_{a=1}^{n}  f(\x^{a})} e^{  \sum\limits_{a=1}^{n}\int_{0}^{\lambda_{a}}R(-w)dw } \Biggr) .
 \end{equation}
Since the additive exponential terms of order $\circ (N)$ have no effect on the results when taking saddle point integration in the limiting regime as $N\rightarrow \infty$  due to the factor $\frac{1}{N}$ outside the logarithm in \eqref{eq:R-transform without zero order} any such terms are dropped further for notational simplicity as in  \cite{BRAR}.

 In order to find the summation in   \eqref{eq:R-transform without zero order}   we employed the procedure in  \cite{BRAR} and  the $nN$ dimensional space spanned by the replicas is split into subshells, defined through $ n \times n$ matrix $\bf{Q}$
\begin{equation}\label{subshels}
S(\Q)= \{  \x^{1},...,\x^{n}\mid ( \x^{0}-\x^{a})^{T}( \x^{0}-\x^{b}) =\frac{N}{\kappa_{n}}Q_{ab} \}.
\end{equation}
The limit $N\rightarrow \infty$ able us to use saddle point integration. Hence we can have the following general result as similar to \cite{BRAR}  but extended in this work with the term, which pertains to CS,  where  we have given the expression that helps to evaluate the performances of the CS estimators using equation \eqref{eq:definition of energy}.

\begin{proposition}
\emph{}
\label{prop: The limiting energy }
The energy $\mathcal{ E }$ from \eqref{eq:definition of energy}, for any inverse temperature $\beta$, any structure of $\Q$ consistent with \eqref{subshels}, and any R-transform $R(.)$ such that $R(\Q)$ is well-defined, is given by 
\begin{equation}\label{eq:the limiting energy}
\bar{\mathcal{ E }}= -\underset{n \rightarrow 0 }{\operatorname{lim }} \frac{1}{ n} Tr[\Q R(-\beta\Q)] ,
\end{equation}
where $\Q$ is the solution to the saddle point equation
\begin{equation}
\Q = \int  \scriptstyle{  \frac{  \sum\limits_{\{ \tilde{\x} \in{ \chi}^{n} \}}  (x^{0}\1-  \tilde{\x}) (x^{0}\1-  \tilde{\x})^{T}   e^{   (x^{0}\1-  \tilde{\x})^{T}\tilde{\Q} (x^{0}\1-  \tilde{\x}) - \frac {\beta\gamma} {\sigma_{u}^2}  \tilde{\x} } }{  \sum\limits_{\{ \tilde{\x} \in {\chi}^n \}} e^{   (x^{0}\1-  \tilde{\x})^{T}\tilde{\Q} (x^{0}\1-  \tilde{\x}) - \frac {\beta\gamma} {\sigma_{u}^2} \tilde{\x} } } } dF_{X^{0}}(x^{0})
\end{equation}

\end{proposition}

\begin{proof}
See Appendix \ref{app: Prof of propostion The limiting energy }.
\end{proof} 
 

Further, to get specific results we need to assume simple structure onto the $n \times n$ cross correlation matrix $\Q$ at the saddle point. So we assume two different assumptions for the entries of $\Q$ called ansatz:  replica symmetry(RS) and replica symmetric breaking (RSB) ansatz.  Then compare the above limiting energy for the different estimators considered in this paper using the two types of ansatz for the CS system.  That is the main purpose that we want to show in this paper. And we took the structures similar to\cite{BRAR} :
\begin{enumerate}
\item replica symmetry ansatz : 
\begin{equation}\label{eq:RS of Q}
\Q =q_{0} \1_{n\times n}+ \frac{b_{0} } {\beta} \I_{n \times n}
\end{equation}
\item one replica symmetry breaking ansatz : 

\begin{equation}\label{eq:RSB of Q}
\Q=q_{1} \1_{ n\times n} + p_{1} \I_{ \frac{n\beta}{\mu_{1}} \times \frac{n\beta}{\mu_{1}} } \otimes \1_{ \frac{\mu_{1}}{\beta} \times \frac{\mu_{1}}{\beta} }+ \frac{ b_{1} }{\beta} \I_{n \times n}
\end{equation}

\end{enumerate}
Applying these assumptions we found some results as given in the following subsections. In the first subsection we assume the RS ansatz which can be considered as the extension of \cite{RGM}.  In the last two subsections we assume  RSB ansatz as an extension of \cite{BRAR}  to CS.  


\subsection{LASSO estimator with RS ansatz}\label{sec:  LASSO estimator with RS ansatz}
Consider the LASSO estimator given in \eqref{LASSO_MAP}, which is equivalent to the solution of the main unconstrained optimization  problem \eqref{unconstrained optimization  problem} in $l_{1}$ penalized sense. Its performance can be expressed in terms of the limiting energy penalty per component using two macroscopic variables $q_{0}$ and $b_{0} $  given by 
\begin{align}\label{eq:macroscopic_q02} 
 q_{0}&=\int_{\mathbb{R}} \int_{\mathbb{R}}  \Bigl| x^{0}- \Psi_{1} \Bigl|^2 Dz dF_{X^{0}}(x^{0}),\\
\end{align}
\begin{align}\label{eq:macroscopic_b02}
b_{0}&=\frac{1}{f_{0}}  \int_{\mathbb{R}} \int_{\mathbb{R}} \Re \Bigg\{x^{0}- \Psi_{1} z^{*}\Bigg\} Dz dF_{X^{0}}(x^{0}),
\end{align}
where
\begin{equation}
 \Psi_{1} =\arg\min_{x \in \chi} \hspace{1mm}\Bigl| -z f_{0}+2e_{0}(x^{0}-x) - \frac {\gamma} {\sigma_{u}^2}  \Bigr|,
\end{equation}
%

\begin{align}\label{eq:macroscopic_e0}
e_{0}&=\frac{1}  { \sigma_{u}^2}   R\Bigl(\frac{-b_{0} }  { \sigma_{u}^2}\Bigr),
\end{align}

\begin{align}\label{eq:macroscopic_f0}
 f_{0}&=\sqrt{ 2 \frac{ q_{0} }{\sigma_{u}^4} R'\Bigl(\frac{-b_{0} }  { \sigma_{u}^2} \Bigr)}, 
\end{align}
and $Dz $ is refering about integration over Gaussian measure, while $dF_{X^{0}}$ refers to integration over the pdf of $x^{0}$ (See Appendix B).  Under RS ansatz assumptions we then get the following statement.

\begin{proposition}
\emph{}
\label{prop: The limiting energy for the LASSO estimator in RS}
Given the LASSO estimator in \eqref{LASSO_MAP} and  the macroscopic variables $q_{0}$ and $b_{0} $, in addition given the conditions in proposition 1 , the energy in \eqref{eq:the limiting energy} simplifies to 
\begin{align}\label{eq: The limiting energy for the LASSO estimator in RS}
\mathcal{ \bar{E}_{\mbox{rs}}^{\mbox{lasso}} }&=\frac{ q_{0}  }{\sigma_{u}^2}R\Bigl( \frac{-b_{0}} { \sigma_{u}^2} \Bigr) - \frac{b_{0}  q_{0}  }{\sigma_{u}^4}  R'\Bigl( \frac{-b_{0}} { \sigma_{u}^2} \Bigr)
\end{align}

\end{proposition}

\begin{proof}
See Appendix B.

\end{proof}

\subsection{  LASSO estimator with 1RSB ansatz} \label{sec:LASSO estimator estimator with 1RSB ansatz}
Moving to the very purpose of the present paper, we use RSB ansatz instead of RS and we repeat what has been done in the above subsections. The limiting energy in this case involves four macroscopic variables like $b_{1}$, $p_{1}$, $q_{1}$, and $\mu_{1}$, which can be given by  the following fixed point equations as $n\rightarrow 0$ and $\beta  \rightarrow \infty$,  as showed in appendix D, and using the compact notation as in \cite{BRAR}.  Let 
\begin{equation}\label{eq: compact exponential argument}
\Delta (y,z) \equiv  e^{-\mu_{1} \min_{x \in \chi}  - 2 \Re\{(x^{0}- x) (f_{1}z^{*}+g_{1}y^{*}) \}    + e_{1}    (x^{0}- x)^2 -\frac {\gamma} {\sigma_{u}^2}   |x|} , \hspace{5mm} (y,z) \in \Re^2
\end{equation}
and its normalized version 
\begin{equation}\label{eq: compact exponential argument}
\tilde{\Delta} (y,z) = \frac{\Delta (y,z) }{\int_{\mathbb{C}}\Delta (\tilde{y},z) d\tilde{y} }  
\end{equation}

\begin{align}\label{eq: macroscopic variables or the 1RSB ansatz1}
b_1 + p_1\mu_1 &=  \frac{1}{f_1} \int  \int_{\mathbb{C}^2} \Re \Big\{  x^{0} - \bigl(\Psi_{2} \bigr)z^* \Big\} \tilde{\Delta} (y,z)\mbox{$D y  Dz  dF_{X^{0}}(x^{0}) $}\\
b_1 +(q_1 +p_1)\mu_1 &=  \frac{1}{g_1} \int \int_{\mathbb{C}^2} \Re \Big\{  x^{0} - \bigl(\Psi_{2}\bigr)y^* \Big\}  \tilde{\Delta} (y,z) \mbox{$D y  Dz  dF_{X^{0}}(x^{0}) $}
\end{align}
\begin{align}\label{eq: macroscopic variables or the 1RSB ansatz2}
q_1 +p_1&=  \frac{1}{g_1} \int  \int_{\mathbb{C}^2}  | \Psi_{2} |^2 \tilde{\Delta} (y,z) \mbox{$D y  Dz  dF_{X^{0}}(x^{0}) $}
\end{align}
where$$\Psi_{2}=\arg\min_{x \in \chi} \hspace{1mm}\Bigl| -(f_{1}z^{*}+g_{1}y^{*})    +  e_{1}   (x^{0}- x) -\frac {\gamma} {\sigma_{u}^2}    \Bigr|,$$
%
and 
\begin{align}\label{the saddel point integration calculations-mu3}
\int_{\frac{b_{1}}{\sigma_{u}^2}}^{\frac{b_{1}+\mu_{1} p_{1}}{\sigma_{u}^2}  }R(-w)dw    &= - R\bigl( -\frac{b_{1}+\mu_{1} p_{1}}{\sigma_{u}^2}\bigr)-\mu_{1}^2 \Bigl((q_{1}  +p_{1} ) g_{1}^2+p_{1}  f_{1}^2\Bigr)  \nonumber\\
& + \int  \int_{\mathbb{C}}    \log \Bigl( \int_{\mathbb{C}} \Delta (y,z) \mbox{$Dy$} \Bigr)\mbox{$  Dz  dF_{X^{0}}(x^{0}) $},
\end{align}
where the other variables $e_{1}$, $f_{1}$, and $g_{1}$, are given by  
\begin{align}\label{eq: macroscopic variables or the 1RSB ansatz}
e_{1}&=\frac{1}{\sigma_{u}^2}R(\frac{-b_{1} } { \sigma_{u}^2} ),\\
g_{1}&=\sqrt{\frac{1}{\mu_{1}\sigma_{u}^2}      \Biggl[R(\frac{-b_{1} } { \sigma_{u}^2} )-R(\frac{-b_{1}-\mu_{1}p_{1} } { \sigma_{u}^2} )\Biggr]},\\
f_{1}&\displaystyle_{\longrightarrow }^{n\rightarrow 0 }   \frac{ 1}{\sigma_{u}^2} \sqrt{ q_{1} R'(\frac{-b_{1}-\mu_{1}p_{1} } { \sigma_{u}^2} )}
\end{align}
Then the following two statements are the extensions of the propositions in \cite{BRAR} to CS problems. 

\begin{proposition}
\emph{}
\label{prop: The limiting energy for the LASSO estimator in RSB  }
Given the LASSO estimator in \eqref{LASSO_MAP} and suppose the random matrix $\J$ satisfies the decomposability property \eqref{eq:decomposiblity property}. Then under some technical assumptions, including one-step replica symmetry breaking, and the  macroscopic variables given by the above fixed point equations,  the effective energy penalty per component converges in probability as $N$, $M \rightarrow \infty$, $N/M < \infty$ , to  

\begin{align}\label{equ: The limiting energy for the LASSO estimator in RSB  }
\mathcal{ \bar{E}_{\mbox{1rsb}}^{\mbox{LASSO}} }&=\scriptstyle{\frac{  1 } { \sigma_{u}^2} (q_{1} + p_{1} +\frac{b_{1}}{\mu_{1}})R(\frac{-b_{1}-\mu_{1} p_{1}  } { \sigma_{u}^2})  -\frac{b_{1}}{\mu_{1}\sigma_{u}^2}R(-\frac{b_{1}} {\sigma_{u}^2}) }  \nonumber\\ &  \scriptstyle{ +   q_{1} (\frac{b_{1}+\mu_{1} p_{1}  } { \sigma_{u}^2} )  R'(\frac{-b_{1}-\mu_{1} p_{1}  } { \sigma_{u}^2})     }  
\end{align}
\end{proposition}

\begin{proof}
See Appendices D.

\end{proof}


\subsection{  Zero-Norm regularizing estimator with 1RSB ansatz}\label{sec:Zero-Norm regularizing estimator with 1RSB ansatz}
The LASSO estimation is considered as the convex relaxation of the Zero-Norm regularizing estimation. Since the latter is a non-convex problem its performance is better evaluated when we use RSB ansatz. 
So extending proposition \eqref{prop: The limiting energy for the LASSO estimator in RSB  }  to this estimator we get the following statement.

\begin{proposition}
\emph{}
\label{prop: Zero-Norm regularizing estimator with 1RSB ansatz}
Given the Zero-Norm regularizing estimator in \eqref{Zero Norm_MAP} and suppose the random matrix $\J$ satisfies the decomposability property \eqref{eq:decomposiblity property}. Then under some technical assumptions, including one-step replica symmetry breaking, the  effective energy penalty per component converges in probablity as $N$, $M \rightarrow \infty$, $N/M < \infty$ , to  
\begin{align}\label{equ:Zero-Norm regularizing with 1RSB ansatz}
\mathcal{ \bar{E}_{\mbox{1rsb}}^{\mbox{zero-norm}} }&=\scriptstyle{\frac{  1 } { \sigma_{u}^2} (q_{1} + p_{1} +\frac{b_{1}}{\mu_{1}})R(\frac{-b_{1}-\mu_{1} p_{1}  } { \sigma_{u}^2})  -\frac{b_{1}}{\mu_{1}\sigma_{u}^2}R(-\frac{b_{1}} {\sigma_{u}^2}) }  \nonumber\\ &  \scriptstyle{ +   q_{1} (\frac{b_{1}+\mu_{1} p_{1}  } { \sigma_{u}^2} )  R'(\frac{-b_{1}-\mu_{1} p_{1}  } { \sigma_{u}^2})     }  
\end{align}
\end{proposition}
\begin{proof}
See Appendix D.

\end{proof}


\section{  Particular Example: Bernoulli-Gaussian Mixture Distribution} \label{sec: Particular Example}

Assume the original vector $\x^{0} \in \mathbb{R}^{N}$ follows a Bernoulli-Gaussian mixture distribution. So following the Bayesian framework analysis in section \ref{sec: Analysis}, let $\x$ be composed of  random variables with each component obeying the pdf 
\begin{equation}\label{equ: Bernolli_Gaussian dist.}
p(x) \sim \left\{ 
\begin{array}{l l}
  \mathcal{N}(0,1) & \quad \mbox{with probability $\rho$}\\
  0 & \quad \mbox{with probability $1-\rho$,}\\ 
\end{array} \right. 
\end{equation}
where $\rho=k/n$, with $k$ being the number of non zero entries of $\x$.  With out loss of generality,  let $\rho=0.1$,  $M/N$ and $k/N$ vary between $0.2$ and $1$.  Also lets assume that the entries of the measurement  matrix $\A$ follow i.i.d. Gaussian random variable of mean zero and variance 1/M. In addition let $\sigma_{u}^2$ be such that the signal to noise ratio is $-10dB$.

We have simulated equations (2.7) and (2.8).  Figure 1 shows MSE  versus $M/N$ of the two estimators, where we se that the $l_{2}$ penalizing estimator, LMMSE, is not as good as the $l_{1}$ penalizing estimator in general.  Figure 2 shows MSE  versus $k/N$ of the two estimators and we see that LMMSE is not sensitive to the sparsity of the vector as compared to the $l_{1}$ penalizing estimator. Note that we have plotted the $l_{1}$-penalizing estimator using different algorithms: LASSO, L1-LS, Log-Bar. 

\begin{figure}[!htbp]
\centering
\includegraphics[height= 0.30\textheight, width=0.80\textwidth]{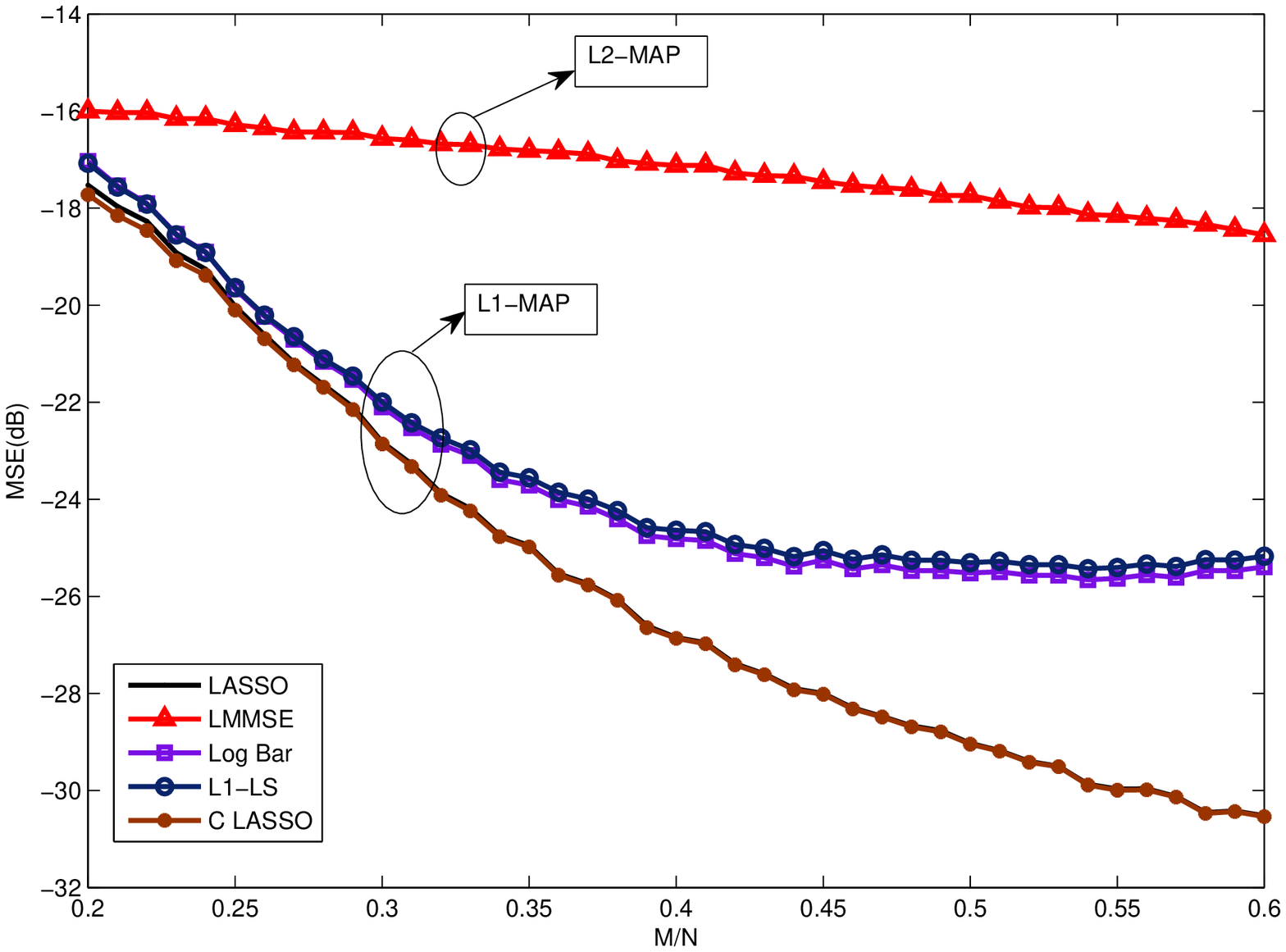}
\caption{\small This figure shows the the normalized mean squared error for the different eastimators in (2.7) and (2.8) versus measurment ration M/N simulated using different algorithms like LASSO, LOG-BAR, L1-LS as L1 penalazing family and LMMSE for the the L2 penalayizing.
\label{fig: the limitting Energy comparison}}
\end{figure}

\begin{figure}[!htbp]
\centering
\includegraphics[height= 0.30\textheight, width=0.80\textwidth]{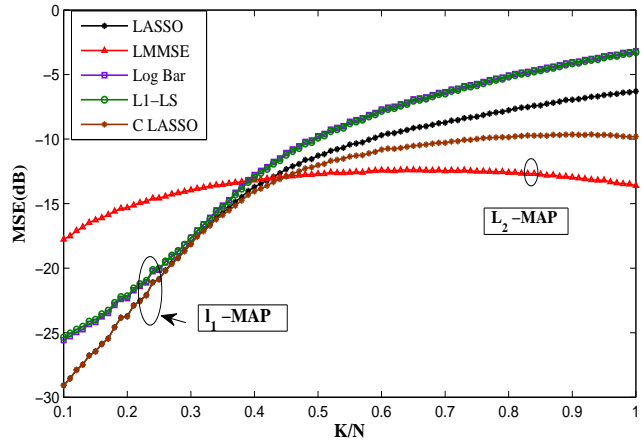}
\caption{\small This figure shows the the normalized mean squared error for the different eastimators in (2.7) and (2.8) versus sparsity ratio k/N simulated using different algorithms like LASSO, LOG-BAR, L1-LS as $l_{1}$ penalazing family and LMMSE for the the $l_{2}$ penalayizing for M=50 and N=100.
\label{fig: the limitting Energy comparison}}
\end{figure}

\begin{figure}[!htbp]
\centering
\includegraphics[height= 0.35\textheight, width=0.80\textwidth]{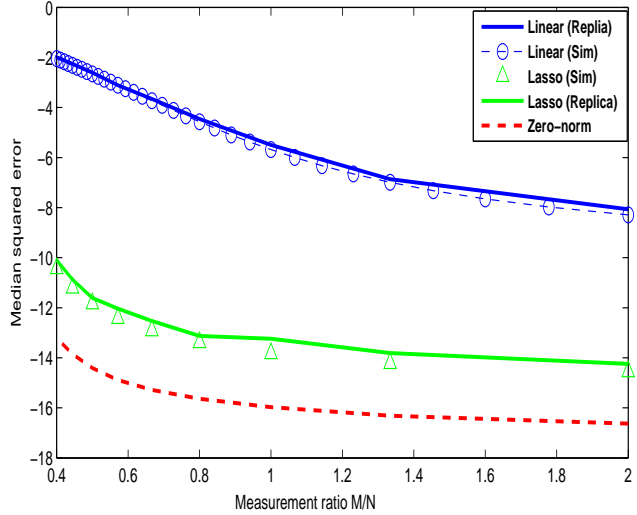}
\caption{\small This figure shows the the  Median squared error against measurment ratio for the eastimators in (2.7)-(2.9) as simulated by Rangan and others \cite{MAP} ploted against M/N instead of N/M and the replica simulation points are included.
\label{fig: the limitting Energy comparison}}
\end{figure}
%
%
%
In both figures, we see that the least square estimator is not good for the compressive sensing problem. In addition, we also observed that simulating the $ l_{0} $ penalizing estimator is hard. However, it is possible to apply statistical physics tools, including replica methods, to analayze the performances of all the estimators mentioned above, including zero norm estimator. 
In \cite{MAP}, median square error was used to compare the different estimators given by \eqref{Linear_MAP}-\eqref{Zero Norm_MAP} as shown here in figure 3. What we do here is that we include 1RSB ansatz analysis of the performance of the CS estimators as each of them are presented here as a MAP estimator. Actually it is one of the conjuctures made by M\"{u}ller and others that the performance of MAP estimators is best done using one step RSB. And we showed it here via the minimized energy expressions as given in the propositions by the equations \eqref{sec:  LASSO estimator with RS ansatz},
\eqref{equ: The limiting energy for the LASSO estimator in RSB  }, and \eqref{equ:Zero-Norm regularizing with 1RSB ansatz}.

\subsection{ Replica symmetry analysis} \label{RS analysis for the example}

Considering the macroscopic variables given by \eqref{eq:macroscopic_q02} and \eqref{eq:macroscopic_b02} and pluging the assumed distributions above and simplyfying it one more step, the fixed point equations become
\begin{align}
 q_{0}&=\frac{\rho^2}{2\pi}\int_{\mathbb{R}} \int_{\mathbb{R}}   \Bigl| \frac{z f_{0} + \frac {\gamma} {\sigma_{u}^2}}{2e_{0}} \Bigr|^2  e^{-\frac{{x^0}^2+z^2}{2}} dz  dx^{0},\\
b_{0}&=\frac{\rho^2}{2\pi}\frac{1}{f_{0}}  \int_{\mathbb{R}} \int_{\mathbb{R}}  \Re \Bigg\{  x^{0}(1- z^{*}) +\Bigl(\frac{z f_{0} +\frac {\gamma} {\sigma_{u}^2}}{2e_{0}} \Bigr) z^{*}\Bigg\}  e^{-\frac{{x^0}^2+z^2}{2}}  dz  dx^{0}.
\end{align}
Using these macroscopic variables  in we find  the  limiting energy numerically which is given under propostion \ref{prop: The limiting energy for the LASSO estimator in RS} and the result is shown in figure \ref{fig: the limitting Energy comparison}. 
%
%


\begin{figure}[h!]
\centering
\includegraphics[height= 0.35\textheight, width=0.85\textwidth]{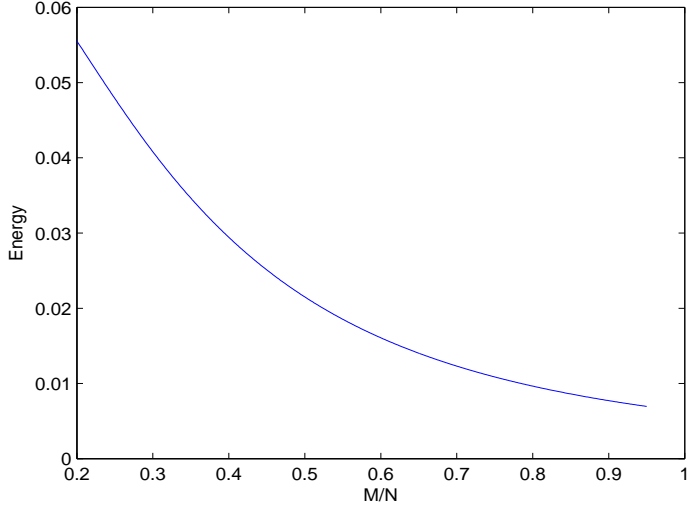}
\caption{\small This figure shows the minimum energy for the error resulting from the RS ansatz for lasso  versus the measurment ratio M/N.
\label{fig: the limitting Energy comparison}}
\end{figure}



\subsection{ Replica symmetry Breaking analysis} \label{RSB analysis for the example}

Considering the same Bernoulli-Gaussian  mixture distribution \eqref{equ: Bernolli_Gaussian dist.}  assumed in this section we consider the macroscopic variables which arises from one step replica symmetry breaking (1RSB)  ansatz. Then the minimum energy per component as $M  \rightarrow \infty, N  \rightarrow \infty$,  while $M/N$ is finite ratio,  which are given by \eqref{equ: The limiting energy for the LASSO estimator in RSB  } and \eqref{equ:Zero-Norm regularizing with 1RSB ansatz} are  dependent up on four macroscopic variables given by \eqref{eq: macroscopic variables or the 1RSB ansatz1} -\eqref{the saddel point integration calculations-mu3}.  The ther first are simplified further as follows:  

We can further simplify \eqref{eq: macroscopic variables or the 1RSB ansatz1}-\eqref{the saddel point integration calculations-mu3} as follows
\begin{align}\label{eq: macroscopic variables or the 1RSB ansatz12}
b_1 + p_1\mu_1 &=  \frac{1}{f_1} \int \int \int_{\mathbb{C}^2} \Re \Big\{ \Bigl( x^{0} - \Psi_{2} \Bigr)z^* \Big\} \mbox{$D y  Dz  dF_{X}(x)dF_{X^{0}}(x^{0}) $}\\
b_1 +(q_1 +p_1)\mu_1 &=  \frac{1}{g_1} \int \int \int_{\mathbb{C}^2} \Re \Big\{ \Bigl( x^{0} - \Psi_{2}\Bigr)z^* \Big\}   \mbox{$D y  Dz  dF_{X}(x)dF_{X^{0}}(x^{0}) $}
\end{align}
\begin{align}\label{eq: macroscopic variables or the 1RSB ansatz21}
q_1 +p_1&=  \frac{1}{g_1} \int \int \int_{\mathbb{C}^2}  | \Psi_{2} |^2  \mbox{$D y  Dz  dF_{X}(x)dF_{X^{0}}(x^{0}) $}
\end{align}

It is possible to simplify these results further and give numerical results. But this is deferred for further work. We expect that the free energy from The RSB ansatz to be greater than the free energy from the RS ansatz for the Zero-Norm regularizing, which can be seen from the analytical terms which have more parameters in \eqref{equ:Zero-Norm regularizing with 1RSB ansatz}. However, for LASSO these free energy, hence the energy error, will be quite similar since for convex minimization problems there is one global minimum and RS ansats is sufficient enough to produce the solution. 
\section{  Conclusion} \label{sec: Conclusion}

In this paper we have used the replica method to analyze the performance of the estimators used in compressed sensing which can be generalized as MAP estimators. And the performance of MAP estimators can well be shown using  replica method including one-step replica breaking ansatz. It is a philosophical standpoint that 1RSB enough to analyze the estimators like MAP. We have only showed here for one particular example for the CS problem, i.e. for Bernoulli-Gaussian distribution. One may be interested to verify it using different examples. In addition we have only compared the estimators performance based on the free energy, but one can also use other metrics such as comparing the input/out put distribution using replica analysis as it is done in \cite{BRAR}. The main result of this paper is analytical analysis for the performance of the estimators used in CS and many things can be extended including efficient algorithms in implementing the numerical analysis. 
\section{Acknowledgments}

We are grateful to Lars  Lundheim, Rodrigo Vicente de Miguel  and Benjamin M. Zaidel for interesting discussions and suggestions.  

\appendix
\section{Important Definitions } \label{Important Definitions }
\subsection{Green's function}
In Classical probability theory (CPT) one is concerned with the densities, moments and comulants of elements of  random matrices. Where as in Random matrix theory (RMT)  also called (Free Random Variable calculus), one is engaed in finding the spectral densities, moments and cumilants ( By Professor Maciej A. Novak).  As Fourier transfom is the generating function for the moments in CPT, Green's function ( also called Stieltjes transform) is the generating function for the spectral moments defined as 
\begin{equation}\label{def: definition of Greens function}
G(z)\equiv \frac{1}{N}  \langle  \mbox{Tr} \frac{1}{z\1_{N } -\X } \rangle \equiv \int \frac{\rho(\lambda) }{z-\lambda}d \lambda \equiv \sum\limits_{n=0}^{\infty} \frac{1}{z^{n+1}}M_{n},
\end{equation}
where $\X$ is $N \times N$ random matrix and  $\1_{N }$ is of the same size unit matrix,  $\lambda$ are the eigenvalues, and $M_{n}$ is the spectral moment. The integral is over the support set of the eigenvalues.
\subsection{R-transform}
The generating function for the cumulants of the CPT is given by the logarithm of the Fourier transfom. In similar maner to the above section we can define the generating function for spectral cumulants. It is called the R-transform (Voiculescu,1986). It is given by 

\begin{equation}\label{def: definition of R-transform}
R(z)\equiv  \sum\limits_{n=1}^{\infty} C_{n}z^{n-1},
\end{equation}
where $C_{n}$ are the spectral cumulants of the random matrix $\X $. We can relate R-transform with Greens's function as follows: 
\begin{equation}\label{def: definition of Greens function}
G(R(z)+\frac{1}{z}) =z  \hspace{5mm} \mbox{or} \hspace{5mm} R(G(z)) + \frac{1}{G(z)} =z.
\end{equation}
The spectral density of the matrix $\J=\A^{T}\A$ converges almost surely to the Marchenko-Pastur law  as $M=\alpha N \rightarrow \infty$ \cite{RGM}.  And the R-transform of this matrix is given by 

\begin{equation}\label{def: definition of R-transform}
R(z)= \frac{1}{1-\alpha z}
\end{equation}
and its derivative with respect to z becomes 
\begin{equation}\label{def: definition of R-transform}
R'(z)= \frac{\alpha}{(1-\alpha z)^2},
\end{equation}
where $\alpha=N/M$ is system load.
\section{Proof of propostion \ref{prop: The limiting energy } } \label{app: Prof of propostion The limiting energy }

The avarage energy penality  can be derived from the average free energy given in \eqref{eq:average free energy} 
\begin{align}\label{eq:avarage energy penality}
\bar{\mathcal{E}}&= \underset{\beta \rightarrow  \infty}{\operatorname{lim }} \frac{1}{ \beta}\bar{\mathcal{F}}=- \underset{\beta \rightarrow  \infty}{\operatorname{lim }} \frac{1}{ \beta} \underset{N \rightarrow  \infty}{\operatorname{lim }} \frac{1}{ N}   \underset{\bf{n},\bf{R}}   {\operatorname{E }} \{\log\hspace{1mm} \mathcal{Z} \}  \nonumber \\&=-\underset{\beta \rightarrow  \infty}{\operatorname{lim }} \frac{1}{ \beta}\underset{n \rightarrow 0}{\operatorname{lim }}    \frac{\partial}{ \partial n} \underbrace{  {\rm   \underset{N \rightarrow  \infty}{\operatorname{lim }} \frac{1}{ N}  \log  \underset{\bf{n}, \bf{J} }{\operatorname{E }}\{ ( \mathcal{Z} )^{n}\} } }_{\rm  \Xi_{n}} .
\end{align}
where $\Xi_{n}$ is given by \eqref{eq:R-transform without zero order}. 
Using \eqref{subshels} as the splitting of the space, we get 
 \begin{equation} \label{eq: integration over the changed variable}
 \Xi_{n} = \underset{N \rightarrow \infty}{\operatorname{lim }} \frac{1}{N} log \int_{\mathbb{R}^{(n+1)^2}} e^{N\mathcal{L}} e^{N\mathcal{I}\{\Q \}}e^{-N\mathcal{G}\{\bf{Q}\}}D\bf{Q}
     \end{equation}
where 
\begin{equation}
D\Q=\prod_{a=0}^{n} dQ_{aa}\prod_{b=a+1}^{n} dQ_{ab}
\end{equation}
is the integration measure, 
\begin{align}
\mathcal{G}(\bf{Q})&= \sum\limits_{a=0}^{n}\int_{0}^{\frac{ \beta \gamma} {\sigma_{u}^2} \lambda_{a}(  \Q )}R(-w)dw \\
&= Tr  \int_{0}^{\frac{ \beta \gamma} {\sigma_{u}^2}  \Q }R(-w)dw\\&=\int_{0}^{  \frac{ \beta \gamma} {\sigma_{u}^2}}Tr[\Q R(-w \Q)]dw
 \end{align}
\begin{equation}
\mathcal{L}=-\frac{\beta \gamma}{2N} \sum\limits_{a=0}^{n}  f(\bf{x}^{a}) \hspace {5mm}   and  \end{equation}
 \begin{equation}
e^{N\mathcal{I} \{ \Q\}}=\sum\limits_{\{\bf{x}^{a}\}}\prod_{a=0}^{n} \delta \bigl(   (\x^{0}  -  \x^{a})^T( \x^{0}  -  \x^{a}) - NQ_{aa} )  \bigr)   \prod_{b=a+1}^{n}  \delta   \bigl(   (\x^{0}  -  \x^{a})^T( \x^{0}  -  \x^{b}) - NQ_{ab}   \bigr)
\end{equation}
denotes probability weight of the subshell composed of Dirac-functions in the real line. This procedure is a change of integration variables in multiple dimensions where the integration of an exponential function over the replicas has been replaced by integration over the variables $\Q$. To evaluate $e^{NC}e^{N\mathcal{I}\{Q\}}$ we follow  \cite{Ralf}, \cite{BRAR} and represent the Dirac measure using the Fourier transform as
\begin{equation}
\delta  \Bigl((\x^{0}  -  \x^{b})^T( \x^{0}  -  \x^{a}) - NQ_{ab}\Bigr) =\int_{\mathcal{J}} e^{    \tilde{Q}_{ab}    \Bigl((\x^{0}  -  \x^{b})^T( \x^{0}  -  \x^{a}) - NQ_{ab}   \Bigr) }  \frac{d\tilde{Q}_{ab}}{2\pi},
\end{equation}
where $ a,b=0,1, \cdots, n$ and this gives

  \begin{align}\label{eq: Fourier transform of dirac}
e^{N\mathcal{L}}e^{N\mathcal{I}\{\Q\}}
&=\sum\limits_{\{\bf{x}^{a}\}}\int_{\mathcal{J}^{n^2}} e^{    \sum\limits_{a,b} \tilde{Q}_{ab} \Bigl( (\x^{0}  -  \x^{b})^T( \x^{0}  -  \x^{a}) - NQ_{ab}\Bigr)  }e^{  \frac {-\beta\gamma} {\sigma_{u}^2}  \sum\limits_{a=1}^{n}  f(\bf{x}^{a})} \tilde{D}\tilde{\Q} \nonumber \\
&=\int_{\mathcal{J}^{n^2}}e^{  -N Tr ( \tilde{\Q}\Q) }\Biggl( \sum\limits_{\{\bf{x}^{a}\}} e^{    \sum\limits_{a,b} \tilde{Q}_{ab} (\x^{0}  -  \x^{b})^T( \x^{0}  -  \x^{a})  }e^{  \frac {-\beta\gamma} {\sigma_{u}^2}  \sum\limits_{a=1}^{n}  f(\bf{x}^{a})}\Biggr) \tilde{D}\tilde{\Q}
\end{align}
where
 \begin{equation}
\tilde{D}\tilde{\Q}=\prod_{a=0}^{n}\Biggl(\frac{d\tilde{Q}_{aa}}{2\pi}\prod_{b=a+1}^{n}\frac{d\tilde{\Q}_{ab}}{2\pi} \Biggr) 
\end{equation}
Assuming  $ f({\x}^{a})=\| {\x}^{a} \|_{1}=    \sum\limits_{i=1}^{N}  |x_{i}^{a} |   $ , which is the sparsity enforcer as described above in LASSO estimator, and after doing some rearrangements, the inner expectation of \eqref{eq: Fourier transform of dirac}  can be given by
\begin{align}
 \sum\limits_{\{\bf{x}^{a}\}}e^{    \sum\limits_{a,b} \tilde{Q}_{ab} (\x^{0}  -  \x^{b})^T( \x^{0}  -  \x^{a}) }e^{  \frac {-\beta\gamma} {\sigma_{u}^2}  \sum\limits_{a=1}^{n}  f(\bf{x}^{a})}
=\prod_{i=1}^{N} \sum\limits_{\{{x}_{i}^{a}\in \chi \}} e^{  (  \sum\limits_{a,b} \tilde{Q}_{ab} ({x}_{i}^{0}  -  {x}_{i}^{b})^T({x}_{i}^{0}  -  {x}_{i}^{a})) - \frac {\beta\gamma} {\sigma_{u}^2}   \sum\limits_{a=1}^{n} \lvert  x_{i}^{a}  \rvert    }
\end{align}
Now defining 
\begin{equation} \label{sum of the exponentials}
M_{i} (\tilde{\Q})= \sum\limits_{\{{x}_{i}^{a}\in \chi \}} e^{  \bigl(  \sum\limits_{a,b} \tilde{Q}_{ab} ({x}_{i}^{0}  -  {x}_{i}^{b})^T({x}_{i}^{0}  -  {x}_{i}^{a})\bigr) - \frac {\beta\gamma} {\sigma_{u}^2}   \sum\limits_{a=1}^{n} \lvert  x_{i}^{a}  \rvert    }
\end{equation} 
we can get  
\begin{equation}
e^{N\mathcal{L}}e^{N\mathcal{I}\{\Q\}}=\int_{\mathcal{J}^{n^2}}e^{  -N Tr ( \tilde{Q}Q)+ \sum\limits_{i=1}^{N} \log M_{i} (\tilde{Q})  } \tilde{D}\tilde{Q}.\end{equation}
Following the i.i.d. assumption for the component of the sparse vector $\x$,  and applying the strong law of large numbers as $N \rightarrow \infty$ we get 
\begin{align} \label{eq:applying law of large numbers}
\log M (\tilde{\Q})
&=\frac{1}{N}\sum\limits_{i=1}^{N} \log M_{i} (\tilde{\Q})  \nonumber\\
& \rightarrow  \int \log  \sum\limits_{\{{x}^{a}\in \chi \}} e^{ \sum\limits_{a,b} \tilde{Q}_{ab} ({x}^{0}  -  {x}^{b})^T({x}^{0}  -  {x}^{a}) - \frac {\beta\gamma} {\sigma_{u}^2}   \sum\limits_{a=1}^{n} |x^{a}|   } \prod_{a=0}^{n}dF_{X}(x^{a}) \nonumber\\
&= \int \log   \sum\limits_{\{\x \in {\chi}^n \}} e^{   (x^{0}\1- \tilde{\x})^{T}\tilde{\Q} (x^{0}\1-  \tilde{\x}) - \frac {\beta\gamma} {\sigma_{u}^2}\tilde{\x}}  \prod_{a=0}^{n} dF_{X^{0}}(x^{0})
\end{align}
where, $\tilde{\x}$ is vector of dimention $n$. Next we apply the saddle point integration concept on  the remaining part of \eqref{eq: integration over the changed variable}, i.e.,  as $N \rightarrow \infty $ the integrand will be dominated by the exponential term with maximal exponent.  Hence in \eqref{eq: integration over the changed variable} only the subshell that corresponds to this extremal value of the correlation between the vectors $\{\x^{a}\}$ is relevant for the calculation of the integral.
  \begin{align}\label{saddel point integration place}
&\int_{\mathbb{R}^{n^2}} e^{N\mathcal{L}} e^{N\mathcal{I}\{\bf{Q}\}}e^{-N\mathcal{G}(\bf{Q})}D\Q  \nonumber\\
&=\int_{\mathbb{R}^{n^2}} \Biggl( \int_{\mathcal{J}^{n^2}}e^{  -N  \mbox{Tr}(\tilde{\Q}\Q)+ \sum\limits_{i=1}^{N} \log M_{i} (\tilde{\Q})  } \tilde{D}\tilde{Q}\Biggr)e^{-N\mathcal{G}\{\bf{Q}\}}D\Q
\end{align}
Therefore, at the saddle point we have the following equations with partial derivatives being zero (see the proof in Appendix B of \cite{BRAR}): 
\begin{equation} \label{partial of G(Q) and trace of Q times  Q tilde}
 \frac{\partial}{\partial \Q }\Bigl[\mathcal{G}(\Q) +  \mbox{Tr}(\tilde{\Q}\Q)\Bigr]=\0   \hspace{7mm}   and \end{equation}
\begin{equation}\label{partial of log M and trace of Q times  Q tilde}
 \frac{\partial}{\partial \tilde{\Q}}\Bigl[ \log M (\tilde{\Q}) -  \mbox{Tr}(\tilde{\Q}\Q)\Bigr]=\0 .
\end{equation}
And from the former we get
 \begin{equation}
 \tilde{\Q}=\beta R(-\frac{ \beta \gamma} {\sigma_{u}^2} \Q)
\end{equation}
and from the later, using \eqref{eq:applying law of large numbers} we finally get
\begin{equation}
\Q =  \int \scriptstyle{  \frac{  \sum\limits_{\{ \tilde{\x} \in{ \chi}^{n} \}}  (x^{0}\1-  \tilde{\x}) (x^{0}\1-  \tilde{\x})^{T}   e^{   (x^{0}\1-  \tilde{\x})^{T}\tilde{\Q} (x^{0}\1-  \tilde{\x}) - \frac {\beta\gamma} {\sigma_{u}^2}   \sum\limits_{a=1}^{n} \lvert  x^{a}  \rvert } }{  \sum\limits_{\{ \tilde{\x} \in {\chi}^n \}} e^{   (x^{0}\1-  \tilde{\x})^{T}\tilde{\Q} (x^{0}\1-  \tilde{\x}) - \frac {\beta\gamma} {\sigma_{u}^2}   \sum\limits_{a=1}^{n} \lvert  x^{a}  \rvert } } } dF_{X^{0}}(x^{0})
\end{equation}

\section{Proof of propostion \ref{prop: The limiting energy for the LASSO estimator in RS} }

Taking the same line of taught as we do for $\Q$, we can assume a natural replicated variables for the symmetric correlation matrix $\tilde{\Q}$ and the 1RSB as follows:  
\begin{enumerate}
\item replica symmetry ansatz : 
\begin{equation}  \label{eq: RS for tilde Q }
\tilde{\Q} =\frac{\beta^2 f_{0}^2}{2}\1_{n\times n}-  \beta e_{0}\I_{n \times n}
\end{equation}
\item one replica symmetry breaking ansatz : 
\begin{equation} \label{eq: RSB for tilde Q }
\tilde{\Q}=\beta^2 f_{1}^2 \1_{ n\times n} + \beta^2 g_{1}^2 \I_{ \frac{n\beta}{\mu_{1}} \times \frac{n\beta}{\mu_{1}} } \otimes \1_{ \frac{\mu_{1}}{\beta} \times \frac{\mu_{1}}{\beta} }- \beta e_{1}\I_{n \times n}
\end{equation}
\end{enumerate}
The variables  $q_{0}$, $b_{0}$, $q_{1}$, $p_{1}$,$b_{1}$, $f_{0}$,$e_{0}$,$f_{1}$,$g_{1}$,$e_{1}$, and $\mu_{1}$ are called the macroscopic variables and they are all functions of n. They all can be calculated from the saddel point equations that we shortly will derive.
First let us try to prove propostion \ref{prop: The limiting energy for the LASSO estimator in RS}  using the ansatz in \eqref{eq:RS of Q} and \eqref{eq: RS for tilde Q }.  We do it using equations \eqref{eq:avarage energy penality}, \eqref{trace of Q and Q tilde} and \eqref{saddel point integration place} and we apply the saddelpoint integration rule. What matters most becomes the argument of the exponential in \eqref{saddel point integration place}. So we first find  $\mbox{Tr}(\tilde{\Q}\Q)$, $\mathcal{G}(\Q)$, $\log M  (\Q)$  and in addition we will find the macroscopic parametrs mentioned before since  our limiting energy penality expressions for the different estimators considered in this paper are calculated interms of the macroscopic variables. Hence using \eqref{eq:RS of Q} and \eqref{eq: RS for tilde Q } we get
\begin{equation}\label{trace of Q and Q tilde} 
\mbox{Tr}(\tilde{\Q}\Q)=n(q_{0}+\frac{b_{0}}{\beta})(\frac{\beta^2f_{0}^2}{2}-\beta e_{0})    + \frac{n(n-1)}{2}q_{0}\beta^2f_{0}^2
\end{equation}
and  using \eqref{sum of the exponentials} and \eqref{eq: RS for tilde Q } again we get
\begin{align}\label{sum of the exponentials1}
M_{i} (\tilde{\Q})
&= \sum\limits_{\{{x}_{i}^{a}\in \chi \}} e^{  \Bigl(  \sum\limits_{a,b} \tilde{Q}_{ab} (x_{i}^{0}  -  x_{i}^{b})( x_{i}^{0}  -  x_{i}^{a})\Bigr) -  \frac {\beta\gamma} {\sigma_{u}^2}   \sum\limits_{a=1}^{n} \lvert  x_{i}^{a}  \rvert   } \\ 
&=\sum\limits_{\{{x}_{i}^{a}\in \chi \}} e^{  \frac{\beta^2f_{0}^2}{2}   \Bigl( \sum\limits_{a=1}^{n}( x_{i}^{0}  -  x_{i}^{a})\Bigr) ^2- e_{0}\beta \sum\limits_{a=1}^{n} ( x_{i}^{0}  -  x_{i}^{a})^2  -  \frac {\beta\gamma} {\sigma_{u}^2}   \sum\limits_{a=1}^{n} \lvert  x_{i}^{a}  \rvert   }\\
&=\sum\limits_{\{{x}_{i}^{a}\in \chi \}}   \int_{\mathbb{R}} e^{ \beta \sum\limits_{a=1}^{n} f_{0}\Re\{(x^{0}  - {x}_{i}^{a}) z^{*}\} - e_{0}  ( x_{i}^{0}  -  x_{i}^{a})^2 - \frac {\gamma} {\sigma_{u}^2}   \lvert  x_{i}^{a}  \rvert  }Dz\\
&=  \int \Biggl( \sum\limits_{\{x\in \chi \}} e^{\beta f_{0} \Re\{(x^{0}  - {x}_{i}^{a}) z^{*}\}+  e_{0} \beta (x^{0}  -  x)^2 - \frac {\beta\gamma} {\sigma_{u}^2}   \lvert  x  \rvert } \Biggr)^n Dz.
\end{align} 
From (B.4) to (B.7) we apply completing the square on the exponential of the argument and  the Hubbard-Stratonovich transform,
\begin{equation} \label{eq:Hubbard-Stratonovich transform}
e^{|x|^2}=\int_{\mathbb C} e^{2\Re \{xz^*\} }Dz,  
\end{equation} 
where $Dz$ is Gaussian measure defined as before,  to linearize the exponential argument.  And we finally transformed the problem to a singele integral and a single summation problem. 
To evaluate $\mathcal{G}(\Q)$ we should first  find the eigenvalues of the matrix L(n). Under the RS ansatz the matrix L(n) has three types of eigenvalues:  $ \lambda_{1}=- (\sigma_{u}^2+n\sigma_{0}^2) ^{-1}(b_{0}+n\beta q_{0}), $
$\lambda_{2}=-(\sigma_{u}^2)^{-1} b_{0}$ and $	\lambda_{3}=0$, and the numbers of degeneracy  for each are 1,  n-1,  and  N-n,  respectively.Thus we get 
\begin{equation}\label{G(Q)}
\mathcal{G}(\Q)=\int_{0}^{\frac{(b_{0}+n\beta q_{0})} { \sigma_{u}^2+n\sigma_{0}^2} }R(-w)dw + (n-1)\int_{0}^{\frac{b_{0}}{\sigma_{u}^2}       }R(-w)dw
\end{equation}
The integral in \eqref{saddel point integration place} is dominated by the maximum argument of the exponential function. Therefore,  the derivative of
 \begin{equation}\label{eq:G(Q) and Tr Q Qtilde}
\mathcal{G}(\Q) + \mbox{Tr}(\tilde{\Q}\Q)
\end{equation}
with respect to $q_{0}$ and $b_{0}$   must vanish as $N\rightarrow \infty$. Plugging \eqref{trace of Q and Q tilde} and \eqref{G(Q)} into \eqref{eq:G(Q) and Tr Q Qtilde} and taking the partial derivatives we get  
\begin{equation}
\frac{\beta n }{\sigma_{u}^2+n\sigma_{0}^2} R\Bigl(\frac{-(b_{0}+n\beta q_{0})}  { (\sigma_{u}^2+n\sigma_{0}^2)} \Bigr) + \frac{n(n-1)}{2}\beta^2f_{0}^2 +n\beta (\frac{\beta f_{0}^2}{2}-e_{0})     =0
\end{equation}
\begin{equation}
  \frac{1}{\sigma_{u}^2+n\sigma_{0}^2}R\Bigl(\frac{-(b_{0}+n\beta q_{0})}  { (\sigma_{u}^2+n\sigma_{0}^2)} \Bigr) + \frac{1}{\sigma_{u}^2}(n-1)R\Bigl(\frac{-b_{0}}   {\sigma_{u}^2}\Bigr) +n(\frac{\beta f_{0}^2}{2} -e_{0})=0,
\end{equation}
respectively. After algebraic simplification and solving for $e_{0}$ and $f_{0}$ we get 
\begin{equation}\label{fixed point eq for macroscopic variable e-zero}
e_{0}=\frac{1 }{\sigma_{u}^2} R\Bigl(\frac{-b_{0}}  { \sigma_{u}^2} \Bigr) ,
\end{equation}
\begin{align}\label{fixed point eq for macroscopic variable f-zero}
f_{0}& =\sqrt{ \frac{2}{n\beta } \Biggl[ \frac{1}{\sigma_{u}^2}R\Bigl(\frac{-b_{0}}   {\sigma_{u}^2}\Bigr) - \frac{1 }{\sigma_{u}^2+n\sigma_{0}^2} R\Bigl(\frac{-(b_{0}+n\beta q_{0}) }  { (\sigma_{u}^2+n\sigma_{0}^2)} \Bigr)\Biggr]}.\end{align}
and with the limit for $n\rightarrow 0$ 
\begin{align}\label{equ: fixed point eq for macroscopic variable f-zero}
f_{0}\displaystyle_{\longrightarrow }^{n\rightarrow 0}\sqrt{ \frac{2}{\beta } \Biggl[ \frac{\sigma_{0}^2}{\sigma_{u}^4}R\Bigl(\frac{-b_{0}}   {\sigma_{u}^2}\Bigr) +\frac{\beta q_{0} \sigma_{u}^2+b_{0}\sigma_{0}^2}{\sigma_{u}^6} R'\Bigl(\frac{-b_{0} }  { \sigma_{u}^2} \Bigr)\Biggr]}.\end{align}
By  substituting \eqref{trace of Q and Q tilde} into \eqref{partial of log M and trace of Q times  Q tilde} and doing the partial derivative of 
\begin{align}
& \log M  (e_{0},f_{0}) -  \mbox{Tr}(\tilde{\Q}\Q)   \nonumber\\
\nonumber  &= \int \log   \sum\limits_{\{\tilde{\x} \in {\chi}^n \}} e^{   (x^{0} \1- \tilde{\x})^{T}\tilde{\Q} (x^{0} \1- \tilde{\x}) -\frac {\beta\gamma} {\sigma_{u}^2} \tilde{\x}  } dF_{X^{0}}(x^{0}) \\ &-   \Bigl(n(q_{0}+\frac{b_{0}}{\beta})(\frac{\beta^2f_{0}^2}{2}-\beta e_{0})    + \frac{n(n-1)}{2}q_{0}\beta^2f^2 \Bigr)\\
 \nonumber  &= \int\log  \int \Biggl( \sum\limits_{\{x\in \chi \}}  e^{\beta f_{0}\Re\{(x^{0}  - {x}_{i}^{a}) z^{*}\}+  e_{0} \beta (x^{0}  -  x)^2 - \frac {\beta\gamma} {\sigma_{u}^2}   \lvert  x  \rvert }  \Biggr)^n Dz dF_{X^{0}}(x^{0})\\ &-     \Bigl(n(q_{0}+\frac{b_{0}}{\beta})(\frac{\beta^2f_{0}^2}{2}-\beta e_{0})    + \frac{n(n-1)}{2}q_{0}\beta^2f_{0}^2 \Bigr),
 \end{align}
with respect to $e_{0}$  and $f_{0}$  and equating to zero we get,
\begin{align}\label{fixed point eq for macroscopic variable q-zero}
q_{0}&=-\frac{b_{0}}{\beta}+\int_{\mathbb{R}} \int_{\mathbb{R}}   \frac{ \sum\limits_{\{x\in \chi \}} (x^{0}  -  x)^2 \zeta     } {  \sum\limits_{\{x\in \chi \}}\zeta  } Dz dF_{X^{0}}(x^{0})
 \end{align}
\begin{align}\label{fixed point eq for macroscopic variable b-zero}
b_{0}&=-\beta nq_{0}+ \frac{1}{f_{0}}\int_{\mathbb{R}} \int_{\mathbb{R}}  \frac{  \sum\limits_{\{x\in \chi \}} \Re\{(x^{0}  - {x}_{i}^{a}) z^{*}\} \zeta } {\sum\limits_{\{x\in \chi \}}  \zeta  } Dz dF_{X^{0}}(x^{0})
 \end{align}
where
\begin{equation}
\zeta=   e^{\beta f_{0}\Re\{(x^{0}  - {x}_{i}^{a}) z^{*}\}+  e_{0} \beta (x^{0}  -  x)^2 - \frac {\beta\gamma} {\sigma_{u}^2}   |x| }.
\end{equation}
 So collecting the  macroscopic variables in \eqref{fixed point eq for macroscopic variable e-zero}, \eqref{fixed point eq for macroscopic variable f-zero}, \eqref{fixed point eq for macroscopic variable q-zero} and \eqref{fixed point eq for macroscopic variable b-zero} and sending $n \rightarrow 0$ we have 
\begin{align}
e_{0}&=\frac{1 }{\sigma_{u}^2} R\Bigl(\frac{b_{0}}  { \sigma_{u}^2} \Bigr) 
 \end{align}
\begin{align}\label{fp equ: f0}
 f_{0}& \displaystyle_{\longrightarrow }^{n\rightarrow 0}\sqrt{ \frac{2}{\beta } \Biggl[ \frac{\sigma_{0}^2}{\sigma_{u}^4}R\Bigl(\frac{-b_{0}}   {\sigma_{u}^2}\Bigr) +\frac{ \beta q_{0} \sigma_{u}^2+b_{0}\sigma_{0}^2}{\sigma_{u}^6} R'\Bigl(\frac{-b_{0} }  { \sigma_{u}^2} \Bigr)\Biggr]}
\end{align}
\begin{align}
q_{0}&=-\frac{b_{0}}{\beta}+\int_{\mathbb{R}} \int_{\mathbb{R}}   \frac{ \sum\limits_{\{x\in \chi \}} (x^{0}  -  x)^2 \zeta     } {  \sum\limits_{\{x\in \chi \}}\zeta  } Dz dF_{X^{0}}(x^{0}),
\end{align}
\begin{align}
b_{0}&\displaystyle_{\longrightarrow }^{n\rightarrow 0}  \frac{1}{f_{0}}\int_{\mathbb{R}} \int_{\mathbb{R}}  \frac{  \sum\limits_{\{x\in \chi \}} \Re\{(x^{0}  - x) z^{*}\} \zeta } {\sum\limits_{\{x\in \chi \}}  \zeta  } Dz dF_{X^{0}}(x^{0}).
\end{align}
And the fixed point equations \eqref{fp equ: f0},  \eqref{fixed point eq for macroscopic variable q-zero}  and \eqref{fixed point eq for macroscopic variable b-zero} further can be simplified via the saddle point integration rule in the limit $\beta \rightarrow \infty$ as 
\begin{align}
 f_{0}&=\sqrt{ 2 \frac{ q_{0} }{\sigma_{u}^4} R'\Bigl(\frac{-b_{0} }  { \sigma_{u}^2} \Bigr)}\\
 q_{0}&=\int_{\mathbb{R}}  \int_{\mathbb{R}}  \Bigl| x^{0}- \arg\min_{x \in \chi} \hspace{1mm}\Bigl| -z f_{0}+2e_{0}(x^{0}-x) - \frac {\gamma} {\sigma_{u}^2}  \Bigr|\Bigl|^2 Dz dF_{X^{0}}(x^{0}),\\
b_{0}&=\frac{1}{f_{0}}  \int_{\mathbb{R}} \int_{\mathbb{R}}  \Re \Bigg\{x^{0}- \arg\min_{x \in \chi} \hspace{1mm}\Bigl| -z f_{0}+2e_{0}(x^{0}-x) - \frac {\gamma} {\sigma_{u}^2}  \Bigr|z^{*}\Bigg\} Dz dF_{X^{0}}(x^{0}).
\end{align}
\vspace{2mm}
Putting together the results above we have 
\begin{align}\label{the xi function}
 \Xi_{n}&=\mathcal{I}\{Q\}+\mathcal{L}-\mathcal{G}(\bf{Q})  \nonumber\\ 
&=-\mathcal{G}(\bf{Q})+ \log  M(\tilde \Q) -  \mbox{Tr}(\tilde{\Q}\Q) \nonumber\\
\nonumber &=-\int_{0}^{\frac{(b_{0}+n\beta q_{0})} { \sigma_{u}^2+n\sigma_{0}^2} }R(-w)dw - (n-1)\int_{0}^{\frac{b_{0}}{\sigma_{u}^2} }R(-w)dw \\ & +  \log M  (e_{0},f_{0}) - \Bigl(n(q_{0}+\frac{b_{0}}{\beta})(\frac{\beta^2f_{0}^2}{2}-\beta e_{0})    + \frac{n(n-1)}{2}q_{0}\beta^2f_{0}^2\Bigr),
 \end{align}
and  the average free energy becomes
 \begin{align}
\beta\bar{\mathcal{F}}&=-\underset{n \rightarrow 0}{\operatorname{lim }}    \frac{\partial}{ \partial n} \underbrace{  {\rm   \underset{N \rightarrow  \infty}{\operatorname{lim }} \frac{1}{ N}  \log  \underset{\bf{n}, \bf{J} }{\operatorname{E }}\{ ( \mathcal{Z} )^{n}\} } }_{\rm  \Xi_{n}} \\ 
\nonumber &=\underset{n \rightarrow 0}{\operatorname{lim }}     \frac{\partial}{ \partial n} \Biggl\{  \int_{0}^{\frac{(b_{0}+n\beta q_{0})} { \sigma_{u}^2+n\sigma_{0}^2} }R(-w)dw + (n-1)\int_{0}^{\frac{b_{0}}{\sigma_{u}^2}       }R(-w)dw \\ 
& - \log M  (e_{0},f_{0}) + \bigl( n(q_{0}+\frac{b_{0}}{\beta})(\frac{\beta^2f_{0}^2}{2}-\beta e_{0})    + \frac{n(n-1)}{2}q_{0}\beta^2f_{0}^2   \bigr) \Biggr\} \\
\nonumber &=\underset{n \rightarrow 0}{\operatorname{lim }} \Biggl\{ \Bigl[ \frac{-(b_{0}+n\beta q_{0})} { \sigma_{u}^2+n\sigma_{0}^2} \Bigr] R\Bigl( \frac{-(b_{0}+n\beta q_{0})} { \sigma_{u}^2+n\sigma_{0}^2} \Bigr)  \\ 
\nonumber & + \frac{-(b_{0}+n\beta q_{0})} { (\sigma_{u}^2+n\sigma_{0}^2)}\Bigl[ -\frac{\Bigl(\beta q_{0}(\sigma_{u}^2+n\sigma_{0}^2 ) -(b_{0}+n\beta q_{0}) \sigma_{0}^2\Bigr) }{(\sigma_{u}^2+n\sigma_{0}^2)^2} \Bigr]R'\Bigl( \frac{-(b_{0}+n\beta q_{0})} { (\sigma_{u}^2+n\sigma_{0}^2)} \Bigr) \\ &+\int_{0}^{\frac{b_{0}}{\sigma_{u}^2} }R(-w)dw  -    \int_{\mathbb{R}}   \int_{\mathbb{R}}   \frac{ \zeta^n \ln \zeta }{ \zeta^n }  Dz dF_{X^{0}}(x^{0}) \Bigg\}\\
\nonumber &= \frac{-b_{0}} { \sigma_{u}^2}  R\Bigl( \frac{-b_{0}} { \sigma_{u}^2} \Bigr)+  \frac{b_{0}\bigl( \beta q_{0}\sigma_{u}^2 -b_{0}\sigma_{0}^2   \bigr)  }{\sigma_{u}^6}  R'\Bigl( \frac{-b_{0}} { \sigma_{u}^2} \Bigr) \\ &+\int_{0}^{\frac{b_{0}}{\sigma_{u}^2} }R(-w)dw  - \int_{\mathbb{R}} \int_{\mathbb{R}}   \ln \zeta  Dz dF_{X^{0}}(x^{0}).
\end{align}
Coming back to the main goal, the solution for the main unconstrained optimization  problem \eqref{unconstrained optimization problem} is given by the extremum  of \eqref{Hamiltonian  the present system}, it is calculated through the free energy by sending $\beta \rightarrow \infty$  as follows 
 \begin{align} \label{the limiting enery interms of the macroscopic variables}
\mathcal{ \bar{E}_{\mbox{rs}}^{\mbox{lasso}} }&= -\underset{\beta \rightarrow \infty }{\operatorname{lim }} \frac{1}{ \beta} \underset{n \rightarrow 0}{\operatorname{lim }}  \frac{\partial}{ \partial n}\Xi_{n}\\
\nonumber&=\underset{\beta \rightarrow \infty }{\operatorname{lim }} \frac{1}{\beta}\Biggl\{  \frac{-b_{0}} { \sigma_{u}^2}  R\Bigl( \frac{-b_{0}} { \sigma_{u}^2} \Bigr)+  \frac{b_{0}\bigl( \beta q_{0}\sigma_{u}^2 -b_{0}\sigma_{0}^2   \bigr)  }{\sigma_{u}^6}  R'\Bigl( \frac{-b_{0}} { \sigma_{u}^2} \Bigr) +\int_{0}^{\frac{b_{0}}{\sigma_{u}^2} }R(-w)dw \\ &  \hspace{10mm} - \int_{\mathbb{R}} \int_{\mathbb{R}}    \ln \zeta  Dz dF_{X^{0}}(x^{0})  \Biggr\}\\
&=\underset{\beta \rightarrow \infty }{\operatorname{lim }}  R\Bigl( \frac{-b_{0}} { \sigma_{u}^2} \Bigr) \Bigl(  \frac{ q_{0}  }{\sigma_{u}^2} +\frac{b_{0}} {\beta \sigma_{u}^2}     \Bigr)   + \frac{b_{0}  q_{0}  }{\sigma_{u}^4}  R'\Bigl( \frac{-b_{0}} { \sigma_{u}^2} \Bigr) \\&  \hspace{10mm} -   \underset{\beta \rightarrow \infty }{\operatorname{lim }}  \frac{1}{\beta}\Biggl\{   \int_{\mathbb{R}} \int_{\mathbb{R}}    \ln \zeta  Dz dF_{X^{0}}(x^{0}) \Biggr\}\\
&=\frac{ q_{0}  }{\sigma_{u}^2}R\Bigl( \frac{-b_{0}} { \sigma_{u}^2} \Bigr) - \frac{b_{0}  q_{0}  }{\sigma_{u}^4}  R'\Bigl( \frac{-b_{0}} { \sigma_{u}^2} \Bigr) .
 \end{align} 
This proves propostion \ref{prop: The limiting energy for the LASSO estimator in RS}.  And to prove propostion \ref{prop: The limiting energy for the Zero-Norm estimator in RS  } what we need is to use the zero norm regularizing term instead of the L1 norm, i.e. using  $ f({\x}^{a})=\| {\x}^{a} \|_{0} =\frac{k}{N} $   in \eqref{eq: Fourier transform of dirac},   and the result will be as in  \eqref{eq: The limiting energy for the Zero-Norm estimator in RS} which differ from  \eqref{eq: The limiting energy for the LASSO estimator in RS}  through the calculation of the macroscopic varables which depend on the distributions of the components of $\x$.


\section{Proof of propostion \ref{prop: The limiting energy for the LASSO estimator in RSB  }  and  \ref{prop: Zero-Norm regularizing estimator with 1RSB ansatz} }

Turning to LASSO estimator with RSB ansatz we first use \eqref{eq:RSB of Q} and  \eqref{eq: RSB for tilde Q } to get 
\begin{align}
 \mbox{Tr}(\tilde{\Q}\Q)&= n(q_{1}  +p_{1} +\frac{b_{1}}{\beta})(\beta^2f_{1}^2+\beta^2g_{1}^2-\beta e_{1}) \\&+ n(\frac{ \mu_{1} }{\beta}-1)(q_{1}  +p_{1} )(\beta^2g_{1}^2+\beta^2f_{1}^2)  +n(n-\frac{ \mu_{1} }{\beta})q_{1}\beta^2 f_{1}^2.
\end{align}
To evaluate $\mathcal{G}(q_{1} , p_{1}, f_{1}, \mu_{1} )$ we should first  find the eigenvalues of the matrix $\L(n)$. Under the RSB ansatz the matrix $\L(n)$ has four types of eigenvalues:  $ \lambda_{1}=- (\sigma_{u}^2+n\sigma_{0}^2) ^{-1}(b_{1}+\mu p_{1} + \beta n q_{1}), $ $\lambda_{2}=-(\sigma_{u}^2)^{-1} (b_{1}+\mu p_{1})$, $\lambda_{3}=-(\sigma_{u}^2)^{-1} b_{1}$ and $\lambda_{4}=0$, and the numbers of degeneracy  for each are 1,  $n\beta/\mu-1$, $n-n\beta/\mu$,  and  $N-n$,  respectively. Hence
\begin{align}\label{eq: the g function}
\nonumber \mathcal{G}(q_{1} , p_{1}, f_{1}, \mu_{1} )&=\int_{0}^{\frac{b_{1}+\mu_{1} p_{1} + \beta n q_{1}} { \sigma_{u}^2+n\sigma_{0}^2} }R(-w)dw + (\frac{n \beta}{\mu_{1}}-1)\int_{0}^{\frac{b_{1}+\mu_{1} p_{1}}{\sigma_{u}^2}  }R(-w)dw \\ & +  (n-\frac{n \beta}{\mu_{1}})\int_{0}^{\frac{b_{1}}{\sigma_{u}^2}  }R(-w)dw
\end{align}
Further with entries of $\tilde{\Q}$ being RSB ansatz  \eqref{eq:applying law of large numbers} will have more involved terms than the RS ansatzs. i.e. ,
\begin{align}\label{eq:logarithm of M1}
&\log M (q_{1} , p_{1}, f_{1}, \mu_{1} ) \nonumber
\\&= \int \log   \sum\limits_{\{\tilde{\x} \in {\chi}^n \}} e^{   (x^{0}\1- \tilde{\x})^{T}\tilde{\Q} (x^{0}\1- \tilde{\x}) -\frac {\beta\gamma} {\sigma_{u}^2}   \tilde{\x} } dF_{X^{0}}(x^{0}) \nonumber\\ 
&= \int  \log   \sum\limits_{\{\x \in {\chi}^n \}}  e^{\beta^2f_{1}^2 \Bigl|\sum\limits_{a=1}^{n} (x^{0}- x_{a})  \Bigr|^2+ \beta^2g_{1}^2\sum\limits_{l=0}^{\frac{n \beta}{\mu}-1} \Bigl| \sum\limits_{a=1}^{\frac{\mu}{\beta}}    (x^{0}- x_{a+\frac{l \mu_{1}}{\beta}})\Bigr|^2 - \beta e_{1}   \sum\limits_{a=1}^{n} (x^{0}- x_{a})^2 -\frac {\beta\gamma} {\sigma_{u}^2}   \sum\limits_{a=1}^{n} \lvert  x_{i}^{a}  \rvert   }  \nonumber\\& \hspace{50mm}  \cdot dF_{x^{0}}(x^{0}).
\end{align}
Using the Hubbard-Stratonovich transform \eqref{eq:Hubbard-Stratonovich transform} we can express \eqref{eq:logarithm of M1}  as in  (c.f. [\cite{RGM} , (66)- (70)]  ) as follows 
\begin{align}
&\log M (q_{1} , p_{1}, f_{1}, \mu_{1} ) \nonumber
\\  &= \int \log   \sum\limits_{\{\x \in {\chi}^n \}}  \int _{\mathbb{C}}e^{\sum\limits_{a=1}^{n} \bigl[ 2\beta f_{1}\Re\{(x^{0}- x_{a}) z^{*} \}    - \beta e_{1}   | (x^{0}- x_{a})|^2 -\frac {\beta\gamma} {\sigma_{u}^2}   |x_{i}^{a} | \bigr]  + \beta^2g_{1}^2\sum\limits_{l=0}^{\frac{n \beta}{\mu}-1}\bigl| \sum\limits_{a=1}^{\frac{\mu}{\beta}} (x^{0}- x_{a+\frac{l \mu_{1}}{\beta}})\bigr|^2   }  \nonumber\\& \hspace{50mm}  \cdot Dz dF_{X^{0}}(x^{0})\nonumber\\
 &=   \int \log  \int_{\mathbb C} \Biggl[ \int_{\mathbb C}   \Biggl(  \sum\limits_{\{\x \in \chi \}} \mbox{$\cal{K}$} \mbox{( $x$, $y$, $z$)}  \Biggr) ^{\frac{\mu_{1}}{\beta}}  \mbox{$Dy$} \Biggr]^{ \frac{ n \beta}{ \mu_{1} } }  \mbox{$ Dz dF_{X^{0}}(x^{0}) $}
\end{align}
where
\begin{equation}
\mbox{$\cal{K}$}\mbox{( $x$, $y$, $z$)} = \mbox{$e^{ 2\beta \Re\{(x^{0}- x) (f_{1}z^{*}+g_{1}y^{*}) \}    - \beta e_{1}   | (x^{0}- x)|^2 -\frac {\beta\gamma} {\sigma_{u}^2}   |x| }$}.  
\end{equation}  
Due to  \eqref{partial of G(Q) and trace of Q times  Q tilde} the partial dervative of 
 \begin{equation} \label{the saddel point integration calculations}
\mathcal{G}(q_{1} , p_{1}, f_{1}, \mu_{1} )+ \mbox{Tr}(\tilde{\Q}\Q)
\end{equation}
with respect to the macroscopic variables $q_{1}$, $p_{1}$, and $b_{1}$ vanishes as $N \rightarrow \infty$ by definition of the sadel point approximation. And pluging \eqref{eq: the g function}  and \eqref{eq:logarithm of M} in \eqref{the saddel point integration calculations} and calculating the partial derivatives and seting them to zero and after some algebraic manipulation we get the folowing set of equations
\begin{align}
0&= n^2\beta^2f_{1}^2+ n\beta \mu_{1} g_{1}^2-n\beta e_{1} + \frac {n\beta}{\sigma_{u}^2+n\sigma_{0}^2} R(\frac{-b_{1}-\mu_{1}p_{1} - \beta n q_{1}} { \sigma_{u}^2+n\sigma_{0}^2} )\\
\nonumber 0&=n\beta \mu_{1} b_{1}^2+n\beta \mu_{1} g_{1}^2-n\beta e_{1}+\frac {(n\beta-\mu_{1})}{\sigma_{u}^2} R(\frac{-b_{1}-\mu_{1} p_{1} } { \sigma_{u}^2} ) \\&+ \frac {\mu_{1}}{\sigma_{u}^2+n\sigma_{0}^2} R(\frac{-b_{1}-\mu_{1}p_{1} - \beta n q_{1}} { \sigma_{u}^2+n\sigma_{0}^2} )\\ 
0&=n\beta f_{1}^2+n\beta g_{1}^2-ne_{1} +\frac {(n-\frac{n\beta}{\mu_{1}})}{\sigma_{u}^2}     R(\frac{-b_{1}} { \sigma_{u}^2} ) + \frac {(\frac{n\beta}{\mu_{1}}-1)}{\sigma_{u}^2} R(\frac{-b_{1}-\mu_{1} p_{1} } { \sigma_{u}^2} ) \\&+ \frac {1}{\sigma_{u}^2+n\sigma_{0}^2} R(\frac{-b_{1}-\mu_{1}p_{1} - \beta n q_{1}} { \sigma_{u}^2+n\sigma_{0}^2} ).
\end{align}
Solving for $e_{1}$, $g_{1}$, $f_{1}$ we get 
\begin{align}
e_{1}&=\frac{1}{\sigma_{u}^2}R(\frac{-b_{1} } { \sigma_{u}^2} ),\\
g_{1}&=\sqrt{\frac{1}{\mu_{1}} \Biggl[\frac{1}{\sigma_{u}^2}R(\frac{-b_{1} } { \sigma_{u}^2} )-\frac{1}{\sigma_{u}^2}R(\frac{-b_{1}-\mu_{1}p_{1} } { \sigma_{u}^2} )\Biggr]},\\
f_{1}&=\sqrt{\frac{1}{n\beta} \Biggl[\frac{1}{\sigma_{u}^2}R(\frac{-b_{1}-\mu_{1} p_{1} } { \sigma_{u}^2} )-\frac{1}{\sigma_{u}^2+n\sigma_{0}^2}R(\frac{-b_{1}-\mu_{1}p_{1} -n\beta q_{1} } { \sigma_{u}^2+n\sigma_{0}^2} )\Biggr]},
\end{align}
and further with the limits $n \rightarrow 0$
\begin{align}
f_{1}&\displaystyle_{\longrightarrow }^{n\rightarrow 0}\sqrt{\frac{1}{\beta} \Biggl[\frac{\sigma_{0}^2}{\sigma_{u}^4}R(\frac{-b_{1}-\mu_{1} p_{1} } { \sigma_{u}^2} )+\frac{(\sigma_{u}^2\beta q_{1}+\sigma_{0}^2(b_{1}+\mu_{1}p_{1}))}{\sigma_{u}^6}R'(\frac{-b_{1}-\mu_{1}p_{1} } { \sigma_{u}^2} )\Biggr]}.
\end{align}
and as $\beta \rightarrow \infty $ we can simplify it further as
\begin{align}
f_{1}&\displaystyle_{\longrightarrow }^{n\rightarrow 0}\sqrt{ \frac{ q_{1}}{\sigma_{u}^4}R'(\frac{-b_{1}-\mu_{1}p_{1} } { \sigma_{u}^2} )}.
\end{align} 
Also due to \eqref{partial of log M and trace of Q times  Q tilde} the partial derivatives of  $$\log M (q_{1} , p_{1}, f_{1}, \mu_{1} ) - \mbox{Tr}(\tilde{\Q}\Q)$$ with respect to $f_{1}$, $g_{1}$, and $e_{1}$, must also vanish as $N \rightarrow \infty $. This produces the following set of equations while taking $n  \rightarrow 0 $.
\begin{align}\label{eq: app: 1RSB: Equation after partial derivative wrt f}
b_1 + p_1\mu_1 &=  \frac{1}{f_1} \int  \int_{\mathbb{C}^2}  \frac{ \left(  \sum_{x \in \chi}  \mbox{$\cal{K}$} \mbox{( $x$, $y$, $z$)}  \right)^{\frac{\mu_1}{\beta}-1} }{\int_{\mathbb{C}} \left(  \sum_{x \in \chi}   \mbox{$\cal{K}$} \mbox{( $x$, $y$, $z$)} 
\right)^{\frac{\mu_1}{\beta}} \mbox{$D \tilde{y} $} }
\nonumber\\& \hspace{35mm}  \cdot \sum_{x \in \chi}  \Re{ \{xz^*\}}  \mbox{$\cal{K}$} \mbox{( $x$, $y$, $z$)}  \mbox{$D y  Dz dF_{X^{0}}(x^{0}) $}
\end{align}
\begin{align}\label{eq: app: 1RSB: Equation after partial derivative wrt f}
b_1 +(q_1 +p_1)\mu_1 &=  \frac{1}{g_1} \int  \int_{\mathbb{C}^2}  \frac{ \left(  \sum_{x \in \chi}  \mbox{$\cal{K}$} \mbox{( $x$, $y$, $z$)}  \right)^{\frac{\mu_1}{\beta}-1} }{\int_{\mathbb{C}} \left(  \sum_{x \in \chi}   \mbox{$\cal{K}$} \mbox{( $x$, $y$, $z$)} 
\right)^{\frac{\mu_1}{\beta}} \mbox{$D \tilde{y} $} }
\nonumber\\& \hspace{23mm}  \cdot \sum_{x \in \chi}  \Re{ \{xy^*\}}  \mbox{$\cal{K}$} \mbox{( $x$, $y$, $z$)} \mbox{$  Dz dF_{X^{0}}(x^{0}) $}
\end{align}
\begin{align}\label{eq: app: 1RSB: Equation after partial derivative wrt f}
q_1 +p_1&= -\frac{b_1}{\beta}+ \frac{1}{g_1} \int  \int_{\mathbb{C}^2}  \frac{ \left(  \sum_{x \in \chi}  \mbox{$\cal{K}$} \mbox{( $x$, $y$, $z$)}  \right)^{\frac{\mu_1}{\beta}-1} }{\int_{\mathbb{C}} \left(  \sum_{x \in \chi}   \mbox{$\cal{K}$} \mbox{( $x$, $y$, $z$)} 
\right)^{\frac{\mu_1}{\beta}} \mbox{$D \tilde{y} $} }
\nonumber\\& \hspace{23mm}  \cdot \sum_{x \in \chi}  |x|^2  \mbox{$\cal{K}$} \mbox{( $x$, $y$, $z$)}  \mbox{$D y  Dz  dF_{X^{0}}(x^{0}) $}.
\end{align}
In addition when we take the partial derivative of 
 \begin{equation} \label{the saddel point integration calculations-mu1}
\mathcal{G}(q_{1} , p_{1}, f_{1}, \mu_{1} )+ \mbox{Tr}(\tilde{\Q}\Q)-\log M (q_{1} , p_{1}, f_{1}, \mu_{1} )
\end{equation}
with respect of  $\mu_{1}$  is vanishes and yields at the limit as $n \rightarrow 0$
\begin{align}\label{the saddel point integration calculations-mu2}
0&= \frac{1}{\mu_{1}^2}\int_{\frac{b_{1}}{\sigma_{u}^2}}^{\frac{b_{1}+\mu_{1} p_{1}}{\sigma_{u}^2}  }R(-w)dw    + \frac{ p_{1}}{\mu_{1}^2}  R\bigl( -\frac{b_{1}+\mu_{1} p_{1}}{\sigma_{u}^2}\bigr)+(q_{1}  +p_{1} ) g_{1}^2+p_{1}  f_{1}^2  \nonumber\\
& +\int \int_{\mathbb{C}} \Biggl[\frac{1}{\mu_{1}^2}\log   \Bigl( \int_{\mathbb C}   \Bigl(  \sum\limits_{\{\x \in \chi \}} \mbox{$\cal{K}$} \mbox{( $x$, $y$, $z$)}  \Bigr) ^{\frac{\mu_{1}}{\beta}}  \mbox{$Dy$} \Bigr) \nonumber\\
& - \int_{\mathbb{C}}  \frac{ \left(  \sum_{x \in \chi}  \mbox{$\cal{K}$} \mbox{( $x$, $y$, $z$)}  \right)^{\frac{\mu_1}{\beta}} }{\beta\mu_{1}^2\int_{\mathbb{C}} \left(  \sum_{x \in \chi}   \mbox{$\cal{K}$} \mbox{( $x$, $y$, $z$)} 
\right)^{\frac{\mu_1}{\beta}} \mbox{$D \tilde{y} $} } \cdot \log\Bigl( \sum_{x \in \chi}  \mbox{$\cal{K}$} \mbox{( $x$, $y$, $z$)} \Bigr) \mbox{$Dy$}\Biggr] \nonumber\\
& \hspace{60mm}  \cdot \mbox{$  Dz  dF_{X^{0}}(x^{0}) $}
\end{align}
So as $\beta  \rightarrow \infty$ these fixed point equations can be simplified as follows: 
\begin{align}\label{eq: app: 1RSB: Equation after partial derivative wrt f}
b_1 + p_1\mu_1 &=  \frac{1}{f_1} \int \int_{\mathbb{C}^2} \Re \Big\{ \Bigl( x^{0} - \Psi_{2} \Bigr)z^* \Big\} \mbox{$D y  Dz  dF_{X^{0}}(x^{0}) $}
\end{align}
\begin{align}\label{eq: app: 1RSB: Equation after partial derivative wrt f}
b_1 +(q_1 +p_1)\mu_1 &=  \frac{1}{g_1} \int  \int_{\mathbb{C}^2} \Re \Big\{ \Bigl( x^{0} - \Psi_{2}\Bigr)y^* \Big\}   \mbox{$D y  Dz  dF_{X^{0}}(x^{0}) $}
\end{align}
\begin{align}\label{eq: app: 1RSB: Equation after partial derivative wrt f}
q_1 +p_1&=  \frac{1}{g_1} \int  \int_{\mathbb{C}^2}  | \Psi_{2} |^2 \mbox{$D y  Dz  dF_{X^{0}}(x^{0}) $}
\end{align}
where $$\Psi_{2}=\arg\min_{x \in \chi} \hspace{1mm}\Bigl| 2 \Re\{(x^{0}- x) (f_{1}z^{*}+g_{1}y^{*}) \}    -  e_{1}   | (x^{0}- x)|^2 -\frac {\gamma} {\sigma_{u}^2}   |x| \Bigr|$$
Puting together the results again as in \eqref{the xi function} and doing again the steps (B.34) to (B.38) for the RSB case 
 \begin{align}
\mathcal{ \bar{E}_{\mbox{1rsb}}^{\mbox{lasso}} }&= -\underset{\beta \rightarrow \infty }{\operatorname{lim }} \frac{1}{ \beta} \underset{n \rightarrow 0}{\operatorname{lim }}  \frac{\partial}{ \partial n}\Xi_{n}\\
&=-\underset{\beta \rightarrow \infty }{\operatorname{lim }} \frac{1}{ \beta} \underset{n \rightarrow 0}{\operatorname{lim }}  \frac{\partial}{ \partial n}\{-\mathcal{G}(\Q) -\mbox{Tr}(\tilde{\Q}\Q) + \log  M(\tilde \Q) \}  \\
&=\underset{\beta \rightarrow \infty }{\operatorname{lim }} \frac{1}{ \beta} \underset{n \rightarrow 0}{\operatorname{lim }}  \frac{\partial}{ \partial n} \Biggl\{\int_{0}^{\frac{b_{1}+\mu_{1} p_{1} + \beta n q_{1}} { \sigma_{u}^2+n\sigma_{0}^2} }R(-w)dw + (\frac{n \beta}{\mu_{1}}-1)\int_{0}^{\frac{b_{1}+\mu_{1} p_{1}}{\sigma_{u}^2}  }R(-w)dw \nonumber\\
 &  +  (n-\frac{n \beta}{\mu_{1}})\int_{0}^{\frac{b_{1}}{\sigma_{u}^2}  }R(-w)dw  +\Bigl[ n(q_{1}  +p_{1} +\frac{b_{1}}{\beta})(\beta^2f_{1}^2+\beta^2g_{1}^2-\beta e_{1}) \nonumber \\
&+ n(\frac{ \mu_{1} }{\beta}-1)(q_{1}  +p_{1} )(\beta^2g_{1}^2+\beta^2f_{1}^2)  +n(n-\frac{ \mu_{1} }{\beta})q_{1}\beta^2 f_{1}^2  \Bigr] \nonumber\\ 
&-\log M (q_{1} , p_{1}, f_{1}, \mu_{1} )\Biggr\}\\
&=\underset{\beta \rightarrow \infty }{\operatorname{lim }} \frac{1}{ \beta} \Biggl\{   (\frac{b_{1}+\mu_{1} p_{1}  } { \sigma_{u}^2}) R(\frac{-b_{1}-\mu_{1} p_{1}  } { \sigma_{u}^2}) \\
 &+  (\frac{b_{1}+\mu_{1} p_{1}  } { \sigma_{u}^2})\frac{ (\beta q_{1} \sigma_{u}^2 -(b_{1}+\mu_{1} p_{1} )\sigma_{0}^2)} { \sigma_{u}^4}  R'(\frac{-b_{1}-\mu_{1} p_{1}  } { \sigma_{u}^2})  \nonumber  \\
&+ \frac{\beta}{\mu_{1}}\int_{0}^{\frac{b_{1}+\mu_{1} p_{1}}{\sigma_{u}^2}  }R(-w)dw +  (1-\frac{ \beta}{\mu_{1}})\int_{0}^{\frac{b_{1}}{\sigma_{u}^2}  }R(-w)dw  \nonumber\\
&+\Bigl[ b_{1}(\beta f_{1}^2+\beta g_{1}^2- e_{1}) + \mu_{1} (q_{1}  +p_{1} )(\beta f_{1}^2+\beta g_{1}^2-\frac{\beta}{\mu_{1}} e_{1}) - \mu_{1} q_{1}\beta f_{1}^2  \Bigr]   \nonumber\\
 &     - \frac{\beta}{\mu_{1}} \int \log  \int_{\mathbb C}  \int_{\mathbb C}   \Biggl(  \sum\limits_{\{\x \in \chi \}} \mbox{$\cal{K}$} \mbox{( $x$, $y$, $z$)}  \Biggr) ^{\frac{\mu_{1}}{\beta}}  \mbox{$Dy$}   \mbox{$ Dz  dF_{X^{0}}(x^{0}) $}\Biggr \}  \\
&= \frac{  q_{1} } { \sigma_{u}^2} (\frac{b_{1}+\mu_{1} p_{1}  } { \sigma_{u}^2}) R(\frac{-b_{1}-\mu_{1} p_{1}  } { \sigma_{u}^2})   + \frac{1}{\mu_{1}}\int_{0}^{\frac{b_{1}+\mu_{1} p_{1}}{\sigma_{u}^2}  }R(-w)dw -  \frac{ 1}{\mu_{1}}\int_{0}^{\frac{b_{1}}{\sigma_{u}^2}  }R(-w)dw  \nonumber \nonumber \\
&+\Bigl[ (b_{1}+ \mu_{1} (q_{1}  +p_{1} ))(f_{1}^2+ g_{1}^2)- e_{1}(q_{1}  +p_{1} ) - \mu_{1} q_{1} f_{1}^2  \Bigr] \nonumber\\
&  -\underset{\beta \rightarrow \infty }{\operatorname{lim }} \frac{1}{ \beta} \Biggl\{\frac{\beta}{\mu_{1}}\int  \log  \int_{\mathbb C}  \int_{\mathbb C}   \Biggl(  \sum\limits_{\{\x \in \chi \}} \mbox{$\cal{K}$} \mbox{( $x$, $y$, $z$)}  \Biggr) ^{\frac{\mu_{1}}{\beta}}  \mbox{$Dy$}   \mbox{$ Dz  dF_{X^{0}}(x^{0}) $} \Biggr \} \\
%
&= \frac{  1 } { \sigma_{u}^2} (q_{1} + p_{1} +\frac{b_{1}}{\mu_{1}})R(\frac{-b_{1}-\mu_{1} p_{1}  } { \sigma_{u}^2})  -\frac{b_{1}}{\mu_{1}\sigma_{u}^2}R(-\frac{b_{1}} {\sigma_{u}^2}) \nonumber  \\ &    \hspace{10mm}+   q_{1} (\frac{b_{1}+\mu_{1} p_{1}  } { \sigma_{u}^2} )  R'(\frac{-b_{1}-\mu_{1} p_{1}  } { \sigma_{u}^2})\\
&=\scriptstyle{\frac{  1 } { \sigma_{u}^2} (q_{1} + p_{1} +\frac{b_{1}}{\mu_{1}})R(\frac{-b_{1}-\mu_{1} p_{1}  } { \sigma_{u}^2})  -\frac{b_{1}}{\mu_{1}\sigma_{u}^2}R(-\frac{b_{1}} {\sigma_{u}^2}) }  \nonumber\\ &  \scriptstyle{ +   q_{1} (\frac{b_{1}+\mu_{1} p_{1}  } { \sigma_{u}^2} )  R'(\frac{-b_{1}-\mu_{1} p_{1}  } { \sigma_{u}^2})     }  
\end{align}


\end{document}